\documentclass[aps,prx,reprint,superscriptaddress,amsmath,amssymb,nobibnotes,longbibliography]{revtex4-2}
\usepackage{iftex}
\ifPDFTeX\pdfoutput=1\usepackage[utf8]{inputenc}\fi
\usepackage[T1]{fontenc}
\usepackage[english]{babel}
\usepackage{amsthm,mathtools,etoolbox,xcolor}
\usepackage{graphicx,booktabs,tikz}
\usetikzlibrary{positioning}
\tikzset{
 penrose/.style={semithick,every node/.style={inner sep=1.5pt,font=\scriptsize}},
 tensor/.style={draw,semithick,minimum size=8mm,inner sep=0pt,font=\normalsize},
}
\usepackage{hyperref}
\newcommand{\sharedauthornote}{%
 \thanks{\raggedright These authors contributed equally.
 \mbox{Contact authors:} Shihao Ru (\mbox{\href{mailto:shihao.ru.2018@gmail.com}{shihao.ru.2018@gmail.com}})
 and Bikun Li (\mbox{\href{mailto:bikunli@uchicago.edu}{bikunli@uchicago.edu}}).}%
}
\definecolor{linkdarkred}{RGB}{128,0,0}
\hypersetup{colorlinks=true,allcolors=linkdarkred}
\newcommand{\Cbb}{\mathbb{C}}
\newcommand{\Rbb}{\mathbb{R}}
\newcommand{\ACal}{\mathcal{A}}
\newcommand{\ECal}{\mathcal{E}}
\newcommand{\HCal}{\mathcal{H}}
\newcommand{\KCal}{\mathcal{K}}
\newcommand{\MCal}{\mathcal{M}}

\newcommand{\OCal}{\mathcal{O}}
\newcommand{\RCal}{\mathcal{R}}
\newcommand{\fopt}{f^{\mathrm{opt}}}
\newcommand{\Zopt}{Z^{\mathrm{opt}}}
\newcommand{\Fe}{F_{\mathrm{e}}}
\newcommand{\Psucc}{P_{\mathrm{succ}}}
\newcommand{\ket}[1]{\lvert#1\rangle}
\newcommand{\bra}[1]{\langle#1\rvert}
\newcommand{\dket}[1]{\lvert#1\rangle\!\rangle}
\newcommand{\norm}[1]{\lVert#1\rVert}
\newcommand{\opnorm}[1]{\norm{#1}}
\newcommand{\fnorm}[1]{\norm{#1}_{\mathrm{F}}}
\newcommand{\inner}[2]{\langle#1,#2\rangle}
\newcommand{\finner}[2]{\inner{#1}{#2}_{\mathrm{F}}}
\newcommand{\dd}{\operatorname{d}}
\DeclareMathOperator{\id}{id}
\DeclareMathOperator{\tr}{tr}
\DeclareMathOperator{\supp}{supp}
\DeclareMathOperator{\rank}{rank}
\newcommand{\lin}{\mathsf{L}}
\newtheorem{theorem}{Theorem}
\newtheorem{lemma}[theorem]{Lemma}
\newtheorem{corollary}[theorem]{Corollary}
\newtheorem{proposition}{Proposition}[section]

\begin{document}
\title{Generalized Reimpell--Werner Iteration}
\author{Shihao Ru}
\sharedauthornote
\affiliation{School of Electrical and Electronic Engineering, Nanyang Technological University, Singapore 639798, Singapore}
\author{Bikun Li}
\sharedauthornote
\affiliation{Chicago Quantum Institute and Pritzker School of Molecular Engineering, University of Chicago, Chicago, Illinois 60637, USA}
\author{Weibo Gao}
\affiliation{School of Electrical and Electronic Engineering, Nanyang Technological University, Singapore 639798, Singapore}
\author{Liang Jiang}
\affiliation{Chicago Quantum Institute and Pritzker School of Molecular Engineering, University of Chicago, Chicago, Illinois 60637, USA}
\begin{abstract}
Quantum measurements and channels determine how information is extracted, encoded, and transmitted in quantum protocols. Optimizing their performance often requires numerical methods that remain practical as Hilbert space dimensions increase. The Reimpell--Werner iteration offers a practical approach to these tasks through repeated matrix updates that respect the constraints. Here, we generalize this iteration to linear objectives with arbitrary Hermitian cost matrices. We prove that the iterates converge to a global optimum whenever the initialization satisfies suitable support overlap conditions. For each fixed problem, choice of iteration parameters, and admissible initialization, $\OCal(1/\varepsilon)$ iterations suffice asymptotically to bring the objective value within $\varepsilon$ of the optimum. These results provide a rigorous foundation for the iteration and broaden the class of optimization problems to which its convergence guarantees apply.
\end{abstract}
\maketitle

\section{Introduction}
Quantum information protocols rely on measurements to extract information from quantum systems and on channels to manipulate, transmit, or preserve it. Designing effective protocols therefore requires both optimizing these operations and understanding the quantum resources available to support them. In minimum error state discrimination~\cite{Jezek2002}, for example, one seeks the measurement that most accurately identifies a state drawn from a known ensemble. Optimizing entanglement fidelity~\cite{Schumacher1996,FletcherShorWin2007} addresses another central question of how faithfully a channel can preserve a given state and its correlations with a reference system. Optimization also provides a way to quantify the resources present in a state. Examples include the robustness of coherence~\cite{Napoli2016} and the robustness of imaginarity~\cite{Wu2021Imaginarity}. These measure the minimum amount of arbitrary noise needed to turn a quantum state into a resource-free state. Despite their different physical interpretations, these tasks admit formulations as semidefinite programs (SDPs)~\cite{Jezek2002,FletcherShorWin2007,Watrous2018,SkrzypczykCavalcanti2023,SiddhuTayur2022}. Their convex structure allows standard SDP methods to approach a global optimum to a prescribed accuracy under the usual feasibility assumptions~\cite{BoydVandenberghe2004}. As Hilbert space dimensions grow, this motivates the development of iterative methods that exploit the structure of measurements and channels to reduce the computational burden.

Two such well-known methods are the Je\v{z}ek--\v{R}eh\'{a}\v{c}ek--Fiur\'{a}\v{s}ek (JRF) iteration~\cite{Jezek2002} for measurement optimization and the Reimpell--Werner (RW) iteration~\cite{ReimpellWerner2005} for channel optimization. For the positive objectives to which they apply, both update the operators through matrix multiplication and normalization based on a positive square root. These operations preserve positivity and enforce the measurement or channel constraints, giving a direct implementation that can avoid the expense of the large linear systems encountered in interior point methods. The algebraic structure also connects the JRF iteration to the pretty good measurement~\cite{HausladenWootters1994,Tyson2009} and to its generalizations and variants~\cite{Mochon2006,TysonWeighted2009,AudenaertMosonyi2014,ZhouChessaChitambarLeditzky2025,McIrvinMohanSikora2024}. A common foundation for these methods was established by Tyson, who identified both as instances of a directional iteration framework and proved that their objective values are nondecreasing~\cite{Tyson2010}. Boundedness then ensures that these values converge, but their limits need not be globally optimal. Moreover, convergence of objective values alone does not establish convergence of the measurement or channel operators for general initialization.

Reimpell and Werner recognized that their iteration may have fixed points that are not optimal~\cite{ReimpellWerner2005,Reimpell2008}. They therefore proposed testing stability and applying random perturbations. Their global convergence discussion concerns this modified, stabilized procedure, with enough Kraus operators to represent any channel. For the original iteration, without stability tests or perturbations, it is not known whether a general initial input of high rank converges to a global optimum. We show that this convergence is guaranteed if the input satisfies an explicit support condition, which every positive definite input fulfills, including cases where the cost matrix is not invertible.

The JRF convergence question has a similar history. The original work reports successful numerical tests without a general analytic proof~\cite{Jezek2002}. An approximation guarantee is available for the first uniformly initialized JRF iterate~\cite{Tyson2009}, while the directional iteration framework establishes monotonicity~\cite{Tyson2010}. These results control finite step performance and objective values, respectively. Nakahira, Kato, and Usuda establish convergence for linearly independent pure ensembles~\cite{NakahiraKato2015}. Their later work explicitly distinguishes this theorem from the then unresolved general case~\cite{NakahiraUsudaKato2017}. L\"u and Dong subsequently treat arbitrary pure ensembles under a support condition on the initialization, together with a restricted class of mixed ensembles~\cite{LuDong2026}. Whether the arbitrarily initialized JRF iteration converges to a globally optimal measurement for every finite ensemble of mixed states, subject to a suitable initial support condition, has remained open. We resolve this open problem in this work.

A broader range of quantum information tasks also motivates extending the scope of these iterations. Examples range from entanglement detection through witnesses~\cite{Lewenstein2000} to the use of local filtering to concentrate entanglement~\cite{Verstraete2003} or reveal hidden Bell nonlocality~\cite{Hirsch2013}. Games with rewards and penalties provide operational benchmarks for quantum memories~\cite{Yuan2021}, while local energy optimization~\cite{Alhambra2019} suggests heralded extensions in quantum thermodynamics. These settings motivate related formulations in which the input states and tests are fixed and a single successful quantum operation is optimized. The resulting objectives can have cost matrices with both positive and negative eigenvalues, calling for an iteration that accommodates general Hermitian cost matrices.

We generalize the RW iteration to arbitrary Hermitian cost matrices under partial trace inequality constraints and establish a global convergence theorem for the resulting iteration. The theorem covers every initial input satisfying a sufficient support condition, including all positive definite inputs and admissible singular ones. Both the optimization variable and its constraint slack converge to an optimal pair, while the limiting normalizer determines the unique dual optimizer of the normalized problem. For positive semidefinite cost matrices with a positive definite partial trace, the result specializes to the original RW iteration. Global convergence of the JRF iteration then follows by a block diagonal specialization.

The proof certifies optimality by showing that, for every input satisfying the support condition, the limiting objective value matches a feasible dual upper bound. The duality gap therefore closes at the fixed point reached by the iteration.

We also quantify how many iterations are needed to approach the optimum. For a fixed problem, choice of iteration parameters, and admissible initialization, $\OCal(1/\varepsilon)$ iterations suffice asymptotically to bring the objective value within $\varepsilon$ of the optimum. In the best case, this bound improves to $\OCal(\log(1/\varepsilon))$. We also establish iteration bounds for bringing the certificate gap and matrix error below a prescribed tolerance $\varepsilon$.

Section~\ref{sec:optimization} presents the SDP formulation, its applications, and the random matrix setting of Figure~\ref{fig:gue}. Section~\ref{sec:iterations} introduces the generalized RW iteration and its specializations, and relates their computational structure to the runtime comparison in Figure~\ref{fig:timing}. Section~\ref{sec:convergence} proves global convergence, and Section~\ref{sec:complexity} establishes iteration bounds. Section~\ref{sec:discussion} discusses implications and open questions. Appendix~\ref{sec:diagonal-rates} derives geometric convergence for diagonal cost matrices and admissible diagonal seeds. Appendix~\ref{sec:slow-trajectories} constructs a family of slower iterations satisfying the support condition and derives sharp polynomial rates for the objective, certificate, and matrix errors, while Appendix~\ref{sec:numerics} gives the numerical simulation details.

\section{Quantum operations optimization}
\label{sec:optimization}
\subsection{Notation and semidefinite formulation}
We work throughout with finite dimensional complex Hilbert spaces and write $\lin(\HCal)$ for the linear space of linear maps from $\HCal$ to itself. We denote the dimensions of $\HCal_A$ and $\HCal_B$ by $d_A\coloneqq\dim\HCal_A$ and $d_B\coloneqq\dim\HCal_B$, respectively. Tensor factors are ordered as $A\otimes B$. The symbols $X^{\dagger}$ and $X^{\mathsf{T}}$ denote the adjoint and transpose of $X$ in the fixed basis, and $\tr_A$ traces out $A$. For Hermitian operators, $X\succeq0$ and $X\succ0$ denote positive semidefiniteness and positive definiteness. The symbols $\opnorm{X}$ and $\fnorm{X}$ denote the operator and Frobenius norms. Vector norms are Euclidean and are denoted by $\norm{\ket{x}}_2$. We use $\finner{U}{V}\coloneqq\tr(U^\dagger V)$, with its real part as the inner product on real tangent spaces. A trace or determinant with a subspace subscript is taken on that subspace. Square roots are positive square roots. For a singular matrix $A\succeq0$, $A^{-1}$ denotes the Moore--Penrose pseudoinverse.

For a fixed Hermitian cost matrix $C=C^\dagger\in\lin(\HCal_A\otimes\HCal_B)$ and a fixed positive semidefinite (PSD) constraint matrix $D\in\lin(\HCal_B)$, consider the primal SDP
\begin{equation}\label{eq:primal}
 \fopt(C,D)\coloneqq\max_{X\succeq0}\bigl\{\tr(CX)\,:\,\tr_A X\preceq D\bigr\}.
\end{equation}
The cost matrix $C$ and the constraint matrix $D$ have finite operator norms, $\opnorm{C}<\infty$ and $\opnorm{D}<\infty$. 

When $C\succeq0$, an optimizer can always be chosen to satisfy $\tr_A X=D$, although other optimizers may leave slack. 
For a Hermitian cost matrix $C$ that is not PSD, an optimizer need not saturate the constraint. Enforcing equality can therefore change the optimum. For example, if $D\ne0$ and $C=-I_{AB}\coloneqq-I_A\otimes I_B$, the inequality optimum is zero at $X=0$, whereas the equality optimum is $-\tr D<0$. Additionally, for Hermitian $C$, if the constraint in Eq.~\eqref{eq:primal} is replaced by the equality $\tr_A X = D$, an optimizer can be obtained by solving the shifted SDP with a sufficiently large $s$ satisfying $C+sI_{AB}\succ 0$. The optimal value under the equality constraint is then $\fopt(C+sI_{AB},D) - s\tr D$.

The dual problem of Eq.~\eqref{eq:primal}, with $Z\in\lin(\HCal_B)$, is
\begin{equation}\label{eq:dual-general}
 \inf_{Z\succeq0}\bigl\{\tr(DZ)\,:\,I_A\otimes Z\succeq C\bigr\}.
\end{equation}
Both constraints are needed for a Hermitian cost matrix. For primal and dual feasible points, weak duality follows from
\begin{equation}\label{eq:weak-duality}
\begin{aligned}
 \tr(DZ)-\tr(CX)
 &=\tr\bigl[Z(D-\tr_A X)\bigr]\\
 &\quad+\tr\bigl[(I_A\otimes Z-C)X\bigr]\ge0.
\end{aligned}
\end{equation}
Equality of primal and dual objectives therefore certifies optimality. The primal feasible set is compact because $\tr X\le\tr D$. Its optimum is attained and nonnegative because $X=0$ is feasible. The dual is strictly feasible at $Z=tI_B$ for $t>\opnorm{C}$, so Slater's theorem gives the same optimal value~\cite{BoydVandenberghe2004}. If $D$ is singular, the dual infimum need not be attained. Any attained dual optimizer is nonunique because adding a nonzero PSD operator supported on $\ker D$ preserves feasibility and the objective.

If $D=0$, only $X=0$ is feasible and $\fopt(C,D)=0$. Otherwise, $\tr_A X\preceq D$ forces every feasible $X$ to be supported on $\HCal_A\otimes\supp(D)$. We therefore replace $\HCal_B$ by $\supp(D)$, compress the cost matrix $C$ to $\HCal_A\otimes\HCal_B$, and restrict $D$ to the new $\HCal_B$. Extension by zero identifies the reduced and original primal feasible sets and preserves their objective values. On the reduced space $D\succ0$, so we define
\begin{subequations}\label{eq:constraint-transformation}
\begin{align}
 C_D&\coloneqq(I_A\otimes D^{1/2})C(I_A\otimes D^{1/2}),\label{eq:transformed-cost}\\
 X_D&\coloneqq(I_A\otimes D^{-1/2})X(I_A\otimes D^{-1/2}),\label{eq:transformed-variable}\\
 Z_D&\coloneqq D^{1/2}ZD^{1/2}\label{eq:transformed-dual}
\end{align}
\end{subequations}
These transformations are invertible and preserve positivity of $X$ and $Z$ on the reduced space. They convert the constraints to $\tr_A X_D\preceq I_B$ and $I_A\otimes Z_D\succeq C_D$, while $\tr(C_DX_D)=\tr(CX)$ and $\tr Z_D=\tr(DZ)$. The original and normalized primal problems therefore have the same optimum. Every normalized dual feasible point gives an upper bound on this optimum through the reduced dual. If $D$ is singular, extending the corresponding reduced dual matrix by zero need not satisfy the original dual constraint.

Thus, apart from the trivial case $D=0$, we set $D=I_B$ without loss of generality for the rest of this work and use $C,X,Z$ for the normalized quantities. The normalized dual is
\begin{equation}\label{eq:dual}
 \min_{Z\succeq0}\bigl\{\tr Z\,:\,I_A\otimes Z\succeq C\bigr\}.
\end{equation}
Both normalized problems are strictly feasible, so their optimal values agree and both are attained. The dual uniqueness assertion in Theorem~\ref{thm:global} concerns this normalized problem.

For the applications below, the input first Choi matrix of a linear map $\Phi\colon\lin(\HCal_A)\to\lin(\HCal_B)$ is
\begin{equation}\label{eq:choi}
 J(\Phi)\coloneqq\sum_{a,a'}\ket a\bra{a'}\otimes\Phi(\ket a\bra{a'}).
\end{equation}
For a completely positive trace preserving (CPTP) map, $J(\Phi)\succeq0$ and $\tr_B J(\Phi)=I_A$. The normalized Choi state is $J(\Phi)/d_A$. We use $J$ for the unnormalized matrix throughout. The adjoint map $\Phi^\dagger\colon\lin(\HCal_B)\to\lin(\HCal_A)$ is defined by $\finner{Y}{\Phi(X)}=\finner{\Phi^\dagger(Y)}{X}$ for $X\in\lin(\HCal_A)$ and $Y\in\lin(\HCal_B)$.

\subsection{Optimizing entanglement fidelity}
Entanglement fidelity quantifies how well an operation preserves an input state and its correlations with a reference system. Let $\rho$ be a density matrix on $\HCal_A$ and let $\ECal\colon\lin(\HCal_A)\to\lin(\HCal_B)$ be a noise channel with Kraus operators $E_\ell$. For a completely positive trace nonincreasing map $\RCal\colon\lin(\HCal_B)\to\lin(\HCal_A)$ with Kraus operators $R_j$, the entanglement fidelity is~\cite{Schumacher1996,FletcherShorWin2007}
\begin{equation}\label{eq:fidelity}
 \Fe(\rho,\RCal\circ\ECal)\coloneqq\sum_{j,\ell}\bigl|\tr(\rho R_jE_\ell)\bigr|^2.
\end{equation}
For a trace nonincreasing map, this expression includes the probability of success. With $\rho$ and $\ECal$ fixed, we seek
\begin{equation}\label{eq:recovery-optimum}
 \max_{\RCal}\Fe(\rho,\RCal\circ\ECal)
\end{equation}
over all such recovery maps.
Set $X_{\RCal}\coloneqq J(\RCal^\dagger)$. This matrix acts on $\HCal_A\otimes\HCal_B$ and satisfies $\tr_A X_{\RCal}=\RCal^\dagger(I_A)\preceq I_B$. Define
\begin{equation}\label{eq:fidelity-cost}
 C_{\rho,\ECal}\coloneqq(\rho^{\mathsf{T}}\otimes I_B)J(\ECal)
                         (\rho^{\mathsf{T}}\otimes I_B).
\end{equation}
A Kraus expansion gives $\Fe(\rho,\RCal\circ\ECal)=\tr(C_{\rho,\ECal}X_{\RCal})$, so Eq.~\eqref{eq:recovery-optimum} is precisely Eq.~\eqref{eq:primal}. The completion described there shows that a CPTP recovery map always attains the optimum. The weighted Choi operator $C_{\rho,\ECal}$ is PSD, and its partial trace over $A$ is
\begin{equation}\label{eq:fidelity-marginal}
 \tr_A C_{\rho,\ECal}=\ECal(\rho^2).
\end{equation}
Positive upper and lower bounds between $\rho$ and $\rho^2$ on their support imply $\ker\ECal(\rho)=\ker\ECal(\rho^2)$. Thus $\tr_A C_{\rho,\ECal}\succ0$ means that the noisy input state $\ECal(\rho)$ has full support on $\HCal_B$. This condition allows both $\rho$ and $C_{\rho,\ECal}$ to be singular. Otherwise the optimization can be restricted to $\supp(\ECal(\rho))$ without changing its optimum.

With $\RCal$ fixed, maximizing $\Fe(\rho,\RCal\circ\ECal)$ over $\ECal$ has the same SDP structure. In this case, the variable is $J(\ECal)$, the constraint is $\tr_B J(\ECal)\preceq I_A$ (the roles of the constrained input and traced output spaces are exchanged), and the PSD cost matrix is $(\rho^{\mathsf{T}}\otimes I_B)J(\RCal^\dagger)(\rho^{\mathsf{T}}\otimes I_B)$. The optimization of encoder and decoder subject to particular conditions plays a key role in approximate quantum error correction~\cite{Li2025Petz,LiJiang2026HighRank}.

\subsection{State discrimination and exclusion}
Consider an ensemble $\{(p_i,\rho_i)\}_{i=1}^m$ on $\HCal_B$, where $p_i>0$, $\sum_i p_i=1$, and each $\rho_i$ is a density matrix. A positive operator valued measure (POVM) has effects $M_i\succeq0$ with $\sum_i M_i=I_B$, and it identifies the state drawn from the ensemble with the success probability
\begin{equation}\label{eq:success-probability}
 \Psucc(\{(p_i,\rho_i)\}_i)\coloneqq\sum_{i=1}^{m} p_i\tr(\rho_iM_i),
\end{equation}
in which the dependence on the POVM is left implicit. We seek
\begin{equation}\label{eq:discrimination-optimum}
 \max_{\{M_i\}_i}\Psucc(\{(p_i,\rho_i)\}_i),
\end{equation}
over all POVMs $\{M_i\}_i$ on $\HCal_B$. Introduce a classical space $\HCal_A\coloneqq\Cbb^m$, define the weighted states $\sigma_i\coloneqq p_i\rho_i$, and set
\begin{align}\label{eq:cq-cost}
 C_{\mathrm{cq}}&\coloneqq\sum_i\ket i\bra i\otimes \sigma_i,\\
 X_M&\coloneqq\sum_i\ket i\bra i\otimes M_i.
\end{align}
The classical-to-quantum channel $\ECal_{\mathrm{cq}}$ encodes an alphabet of size $d_A=m$ by mapping each symbol $i$ to the state $\rho_i$ on $\HCal_B$. The cost matrix $C_{\mathrm{cq}}$ is its prior weighted Choi operator, which reduces to the normalized Choi state $J(\ECal_{\mathrm{cq}})/d_A$ for uniform priors.

The identity $\tr(C_{\mathrm{cq}}X_M)=\Psucc(\{(p_i,\rho_i)\}_i)$ relates state discrimination to Eq.~\eqref{eq:primal} with $D=I_B$. Here the cost matrix $C_{\mathrm{cq}}$ is PSD, and discarding off-diagonal classical blocks preserves feasibility and the objective. The completion argument following Eq.~\eqref{eq:primal}, applied to one diagonal block, therefore yields an optimal normalized POVM with $\tr_A X_M=\sum_iM_i=I_B$.

To satisfy the partial trace condition required for the original RW iteration, we assume that the ensemble average has full support,
\begin{equation}\label{eq:ensemble-marginal}
 \overline{\rho}\coloneqq\sum_i p_i\rho_i=\tr_A C_{\mathrm{cq}}\succ0.
\end{equation}
Since every state has a nonzero prior probability, this assumption means that the supports of the ensemble states together span $\HCal_B$. Individual states need not be full rank or linearly independent. If the ensemble average is singular, we can restrict the problem to its support without changing any discrimination probability.

Explicitly, $\ECal_{\mathrm{cq}}(\ket i\bra j)\coloneqq\delta_{ij}\rho_i$, so its Choi matrix is
\begin{equation}\label{eq:preparation-choi}
 J(\ECal_{\mathrm{cq}})=\sum_i\ket i\bra i\otimes\rho_i.
\end{equation}
If $\RCal_M$ is the measurement channel with effects $M_i$, then $X_M=J(\RCal_M^\dagger)$, consistent with the convention of the preceding subsection.

The same ensemble also defines minimum error state exclusion, in which outcome $i$ rules out state $\rho_i$ rather than identifying it. The pretty bad measurement provides a candidate POVM for this task~\cite{McIrvinMohanSikora2024}, whose exclusion error is $\sum_i\tr(\sigma_iM_i)$. When $m\ge2$, define the complementary cost matrix~\cite{BandyopadhyayJainOppenheimPerry2014}
\begin{equation}\label{eq:exclusion-cost}
 C_{\mathrm{excl}}\coloneqq\frac{I_A\otimes\overline{\rho}-C_{\mathrm{cq}}}{m-1}\succeq0.
\end{equation}
For a normalized POVM, its objective is
\begin{equation}\label{eq:exclusion-objective}
 \tr(C_{\mathrm{excl}}X_M)=\frac{1-\sum_i\tr(\sigma_iM_i)}{m-1}.
\end{equation}
As in state discrimination, discarding off-diagonal classical blocks and completing the POVM preserves optimality. Hence Eq.~\eqref{eq:primal} with $C=C_{\mathrm{excl}}$ gives the minimum exclusion error as $1-(m-1)\fopt$. Moreover, $\tr_A C_{\mathrm{excl}}=\overline{\rho}\succ0$ on the support of the ensemble average, so the original RW iteration applies to this PSD cost matrix.

\subsection{Further applications}
\label{sec:further-applications}
Coherence quantification and local energy extraction illustrate the range of quantum information tasks covered by this SDP framework.

For any density matrix $\rho$ on $\HCal_B$, including singular states, the robustness of coherence $R_{\mathrm{coh}}(\rho)$~\cite{Napoli2016} is the minimum weight of arbitrary noise needed to make its mixture with $\rho$ diagonal in a fixed orthonormal basis $\{\ket i_B\}_i$. To express this quantity as an SDP, let $\HCal_A$ be a copy of $\HCal_B$ and define the isometry $U_{\mathrm{coh}}\ket i_B\coloneqq\ket i_A\otimes\ket i_B$. Writing $\rho_{\mathrm{off}}$ for $\rho$ with its diagonal entries set to zero, we define the Hermitian cost matrix $C_{\mathrm{off}}\coloneqq U_{\mathrm{coh}}\rho_{\mathrm{off}}U_{\mathrm{coh}}^\dagger$. With $\tau$ ranging over density matrices on $\HCal_B$, the robustness satisfies
\begin{subequations}\label{eq:coherence-robustness}
\begin{align}
 R_{\mathrm{coh}}(\rho)
 &\coloneqq\min_{\lambda\ge0,\,\tau}
 \left\{\lambda:\frac{\rho+\lambda\tau}{1+\lambda}\text{ is diagonal}\right\}
 \label{eq:coherence-definition}\\
 &=\min_{\Delta\succeq\rho_{\mathrm{off}}}\bigl\{\tr\Delta:\Delta\text{ is diagonal}\bigr\}
 \label{eq:coherence-majorant}\\
 &=\min_{Z\succeq0}
 \bigl\{\tr Z:I_A\otimes Z\succeq C_{\mathrm{off}}\bigr\}.
 \label{eq:coherence-dual}
\end{align}
\end{subequations}
The substitution $\Delta\coloneqq\rho_{\mathrm{off}}+\lambda\tau$ gives Eq.~\eqref{eq:coherence-majorant}, with $\tr\Delta=\lambda$. Compressing the constraint in Eq.~\eqref{eq:coherence-dual} by $U_{\mathrm{coh}}$ yields a diagonal majorant of $\rho_{\mathrm{off}}$ with trace $\tr Z$. Conversely, every such majorant $\Delta$ is PSD because $\rho_{\mathrm{off}}$ has zero diagonal, and the inequality $I_A\otimes\Delta\succeq U_{\mathrm{coh}}\Delta U_{\mathrm{coh}}^\dagger\succeq C_{\mathrm{off}}$ makes $Z=\Delta$ feasible. This identifies Eq.~\eqref{eq:coherence-dual} with Eq.~\eqref{eq:dual-general} for $D=I_B$ and establishes $R_{\mathrm{coh}}(\rho)=\fopt(C_{\mathrm{off}},I_B)$.

Local energy extraction provides a complementary application, in which the decrease in total energy of a fixed bipartite state is optimized by acting on one subsystem~\cite{Alhambra2019}. For a state $\rho_{BR}$ with Hamiltonian $H_{BR}$, let $\HCal_A$ be a copy of $\HCal_B$ and denote the relabeled Hamiltonian by $H_{AR}$. A local CPTP map $\ECal\colon\lin(\HCal_B)\to\lin(\HCal_A)$ produces $\omega_{\ECal}\coloneqq(\ECal\otimes\id_R)(\rho_{BR})$, so the maximum energy decrease is
\begin{equation}\label{eq:energy-extraction}
 \max_{\ECal}\bigl\{\tr(H_{BR}\rho_{BR})-\tr(H_{AR}\omega_{\ECal})\bigr\}.
\end{equation}
This task admits a heralded extension in which $\ECal$ is completely positive and trace nonincreasing (CPTNI). Writing $e_0\coloneqq\tr(H_{BR}\rho_{BR})$, the payoff becomes $e_0\tr\omega_{\ECal}-\tr(H_{AR}\omega_{\ECal})$. It weights the conditional energy decrease by the success probability and assigns zero payoff to unsuccessful outcomes. Such probability weighting also appears in the specific measurement protocols of Ref.~\cite[Eq.~(9)]{Chaki2025}. Both scores quantify energy reduction without including the energy required to implement the operation.

To connect these formulations to the SDP framework, set $X\coloneqq J(\ECal^\dagger)$ and define the Hermitian cost matrix $C$ by $\tr(CX)=e_0\tr\omega_{\ECal}-\tr(H_{AR}\omega_{\ECal})$ for every $\ECal$. The CPTP constraint is $\tr_A X=I_B$, so the positive shift reduction described after Eq.~\eqref{eq:primal} applies. The heralded version instead satisfies $\tr_A X\preceq I_B$ and is precisely Eq.~\eqref{eq:primal} with $D=I_B$ and a possibly indefinite cost matrix.

To complement these structured applications, Figure~\ref{fig:gue} displays iteration count statistics for the generalized RW iteration with Hermitian cost matrices sampled from the Gaussian unitary ensemble (GUE), with $d_A=d_B=d$. We use Wigner normalization~\cite{Wigner1958}, so that $\mathbb{E}(C^2)=I_{AB}$. Specifically, the diagonal entries are independent centered real Gaussians with variance $1/d^2$, while the real and imaginary parts of the independent upper triangular entries have variance $1/(2d^2)$. Hermitian symmetry determines the remaining entries. Appendix~\ref{sec:gue-details} specifies the sampling, initialization, and error measures.

\section{Iterative algorithms}
\label{sec:iterations}
The optimization problems in the preceding section maximize a linear objective over PSD matrices subject to a partial trace constraint. For a PSD cost matrix, the original RW iteration multiplies the current matrix on both sides by the cost matrix and then normalizes the result to satisfy that constraint with equality. We first generalize this construction to Hermitian cost matrices, then recover the original RW iteration and its JRF specialization for state discrimination.

To make the cost matrix PSD, we introduce a fixed scalar $s\ge0$ and define the shifted cost matrix $\widetilde C\coloneqq C+sI_{AB}$. A sufficiently large shift achieves positivity, but the additional objective term $s\tr X$ can change the optimizer because the inequality constraint does not fix $\tr X$.

To preserve the original optimization problem, we introduce a slack matrix $S\succeq0$ on $\HCal_B$ and use the equivalent constraint~\cite[Eqs.~(3.20)--(3.21)]{Reimpell2008}
\begin{equation}\label{eq:slack-reformulation}
 \tr_A X+S=I_B.
\end{equation}
Every feasible $X$ determines the unique slack $S=I_B-\tr_A X$. Taking the trace gives $\tr X+\tr S=d_B$, so the augmented objective satisfies
\begin{equation}\label{eq:shifted-objective}
\begin{split}
 \tr(\widetilde C X)+s\tr S
 &=\tr(CX)+s(\tr X+\tr S)\\
 &=\tr(CX)+s d_B.
\end{split}
\end{equation}
This objective therefore has exactly the same optimizing matrices $X$. Related spectral shifts for indefinite quadratic objectives appear in fixed norm optimization~\cite[Example~6]{Journee2010} and in factored SDPs with fixed diagonal~\cite{AktasKroer2026}. Here the slack matrix makes the shift compatible with the original inequality constraint. The two cost matrices, $\widetilde C$ and $sI_B$, lead to coupled RW updates for $X$ and $S$, with a common normalizer enforcing Eq.~\eqref{eq:slack-reformulation}.

For this coupled construction, we call the shift \emph{admissible} when
\begin{equation}\label{eq:admissible-shift}
 s\ge0,\qquad \widetilde C\succeq0,\qquad
 \tr_A\widetilde C+sI_B\succ0.
\end{equation}
Equivalently, $s\ge\max\{0,-\lambda_{\min}(C)\}$ and either $s>0$ or $\tr_A C\succ0$. The strict choice $s>\max\{0,-\lambda_{\min}(C)\}$ always suffices, although the displayed conditions also allow a singular $\widetilde C$.

To state the recursion, we use parenthesized superscripts for successive iterates, with ``$(0)$'' denoting the initial value and ``$(\infty)$'' the limit when it exists. From a general seed pair $X^{(0)},S^{(0)}\succeq0$, the generalized RW iteration is defined by
\begin{equation}\label{eq:rw}
\left\{
\begin{aligned}
 Y^{(k)}&\coloneqq\Bigl[\tr_A(\widetilde C X^{(k)}\widetilde C)+s^2S^{(k)}\Bigr]^{1/2},\\
 X^{(k+1)}&\coloneqq(I_A\otimes Y^{(k)})^{-1}\widetilde C X^{(k)}\widetilde C(I_A\otimes Y^{(k)})^{-1},\\
 S^{(k+1)}&\coloneqq s^2\bigl(Y^{(k)}\bigr)^{-1}S^{(k)}\bigl(Y^{(k)}\bigr)^{-1},
\end{aligned}
\right.
\end{equation}
provided $Y^{(k)}\succ0$. We call $Y^{(k)}$ the normalizer because it ensures that the updated PSD matrices satisfy
\begin{equation}\label{eq:pair-normalization}
\begin{split}
 &\tr_A\bigl(X^{(k+1)}\bigr)+S^{(k+1)}\\
 &\qquad=\bigl(Y^{(k)}\bigr)^{-1}\bigl(Y^{(k)}\bigr)^2\bigl(Y^{(k)}\bigr)^{-1}=I_B.
\end{split}
\end{equation}
Consequently, $X^{(k)}$ is feasible for the original inequality problem, $S^{(k)}$ is its constraint slack, and Eq.~\eqref{eq:shifted-objective} applies for $k\ge1$. The initial pair need not be normalized, but the same identity holds at $k=0$ if it is.

Hermitian measurement payoff operators are made PSD by a common operator shift in~\cite{NakahiraUsudaKato2017}. With no additional scalar constraints, we recover the block diagonal specialization of Eq.~\eqref{eq:rw} from the measurement update formula in that work by adjoining an outcome $S$ with zero payoff and shifting every payoff by $sI_B$.

For a current iterate pair $(X^{(k)},S^{(k)})$ that is positive definite and normalized, and a shift $s>\max\{0,-\lambda_{\min}(C)\}$, Eq.~\eqref{eq:rw} can be treated as combining the ascent form of the Bures--Wasserstein gradient step~\cite{Fan2024} with the Bures projection onto a fixed partial trace constraint~\cite[Cor.~14]{AfhamTomamichel2026}, applied after the slack augmentation in Eq.~\eqref{eq:slack-reformulation}.

The normalization ensures feasibility after each defined step, while our convergence guarantee also uses a support condition on the initial pair. Each update is a congruence, so $\rank X^{(k+1)}\le\rank X^{(k)}$, and its limit cannot be optimal if every optimizer has rank greater than $\rank X^{(0)}$. Our sufficient conditions, which guarantee convergence to a global optimum, are
\begin{equation}\label{eq:support-condition}
 \begin{aligned}
 \ker X^{(0)}\cap\supp(\widetilde C)&=\{0\},\\
 \ker S^{(0)}\cap\supp(sI_B)&=\{0\}.
 \end{aligned}
\end{equation}
Any $X^{(0)},S^{(0)}\succ0$ fulfill the above conditions, and the first condition also allows singular $X^{(0)}$ when $\widetilde C$ is singular. The second requires $S^{(0)}\succ0$ when $s>0$ and is vacuous when $s=0$. At a positive shift, setting $S^{(0)}=0$ would keep the slack zero forever and could exclude the inequality optimum. One convenient normalized choice is $X^{(0)}\coloneqq I_{AB}/(d_A+1)$ and $S^{(0)}\coloneqq I_B/(d_A+1)$. Under an admissible shift and Eq.~\eqref{eq:support-condition}, Section~\ref{sec:convergence} proves that every normalizer has an ordinary inverse and that the entire sequence converges to an optimum. These support conditions need not be necessary.

The original RW iteration~\cite{ReimpellWerner2005,Reimpell2008} is a special case of the generalized RW iteration in Eq.~\eqref{eq:rw}, obtained by taking $s=0$ when $C\succeq0$ and $\tr_A C\succ0$. Its recursion from a general PSD seed $X^{(0)}$ is
\begin{equation}\label{eq:rw-psd}
\left\{
\begin{aligned}
 Y^{(k)}&\coloneqq\bigl[\tr_A(CX^{(k)}C)\bigr]^{1/2},\\
 X^{(k+1)}&\coloneqq(I_A\otimes Y^{(k)})^{-1}CX^{(k)}C(I_A\otimes Y^{(k)})^{-1},
\end{aligned}
\right.
\end{equation}
at each step for which $Y^{(k)}\succ0$. No slack matrix is needed, and each defined iterate satisfies $\tr_A X^{(k+1)}=I_B$. Its sufficient support condition is $\ker X^{(0)}\cap\supp(C)=\{0\}$, the zero-shift specialization of Eq.~\eqref{eq:support-condition}. Multiplying the PSD cost matrix or its seed by a positive scalar leaves the first normalized iterate unchanged.

\begin{figure*}[t]
 \centering
 \includegraphics[width=\textwidth]{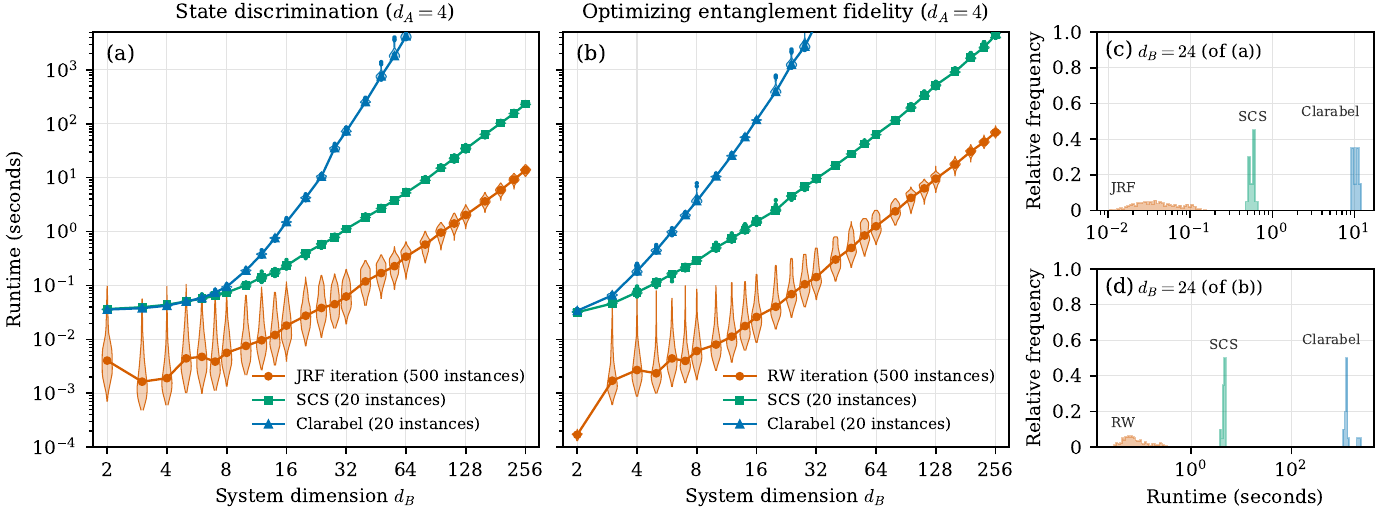}
 \caption{Runtime at fixed $d_A=4$ and objective error at most $10^{-7}$ for (a) state discrimination and (b) entanglement fidelity optimization. Accuracy is assessed against a precomputed numerical dual upper bound whose construction is excluded from the timings. Violins show the density of logarithmic runtime, each scaled to the same maximum width, and connected markers show medians. Each group contains $500$ instances for the original JRF or RW iteration and $20$ for SCS (Splitting Conic Solver) or Clarabel. Individual solver runtimes are also marked. Panels (c) and (d) show the corresponding runtime histograms at $d_B=24$, with relative frequencies normalized separately for each method. Sampling and timing conventions are given in Appendix~\ref{sec:runtime-details}.}
 \label{fig:timing}
\end{figure*}

\begin{figure*}[t]
 \centering
 \includegraphics[width=\textwidth]{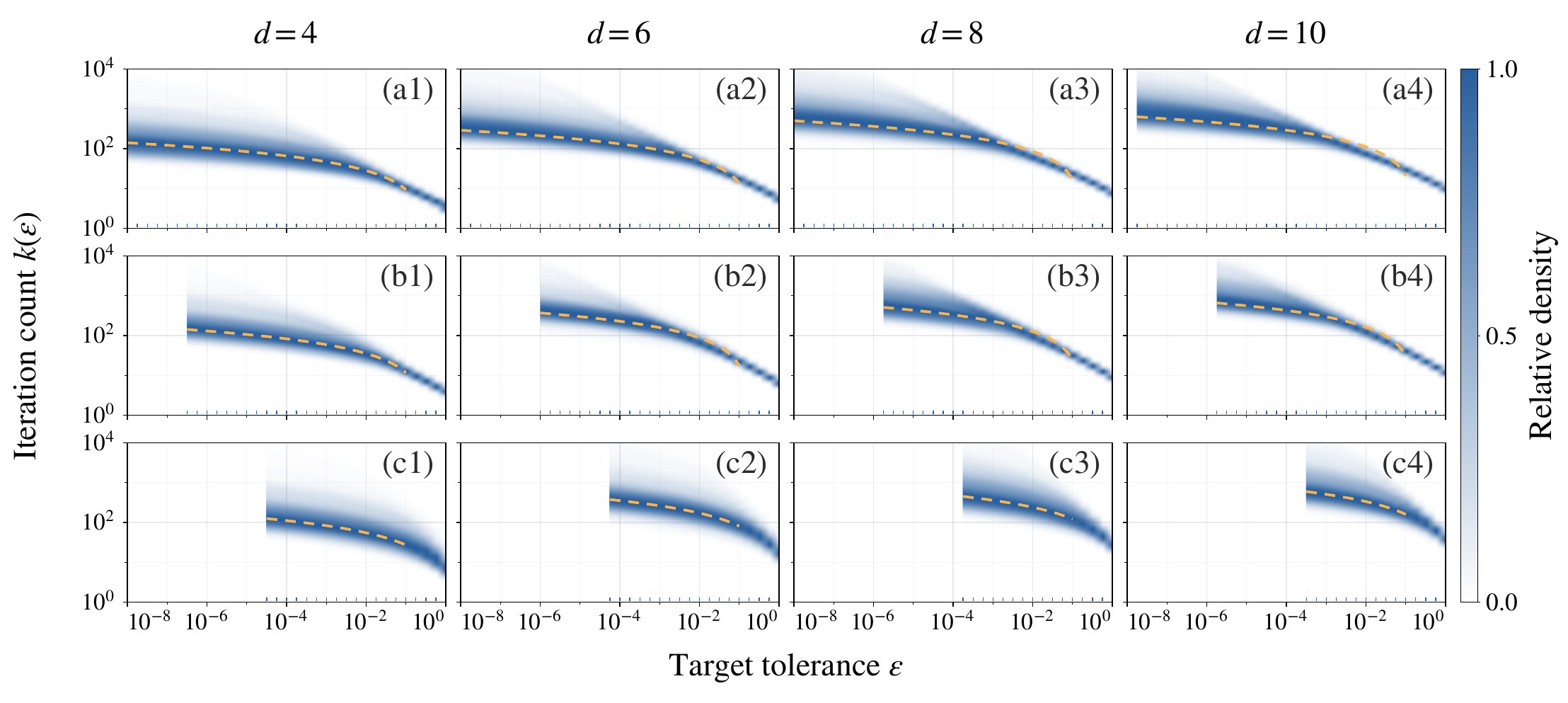}
 \caption{Iteration counts of the generalized RW iteration for Hermitian cost matrices drawn from the GUE, with $d_A=d_B=d$ and $500$ independent instances at each dimension. For a target tolerance $\varepsilon$, $k(\varepsilon)$ is the first iteration at which the error measure of the row, defined in Appendix~\ref{sec:gue-details}, is at most $\varepsilon$. The rows show, from top to bottom, objective discrepancies relative to the stopping value $f^{\mathrm{stop}}$, numerical certificate gaps, and empirical Frobenius distances to refined numerical references. Color shows the kernel density estimate of $\log_{10}k(\varepsilon)$ over the instances, normalized to its peak at each measured target, and ticks along the bottom edge of each panel mark the measured targets. Yellow dashed curves show fits of $c_1\log(1/\varepsilon)+c_2$ to the density peaks, with $c_1$ and $c_2$ fitted separately in each panel.}
 \label{fig:gue}
\end{figure*}

For state discrimination with weighted states $\sigma_i$ as in Eq.~\eqref{eq:cq-cost}, the JRF iteration~\cite{Jezek2002} starts from a seed family $\{M_i^{(0)}\}_{i=1}^m$ with $M_i^{(0)}\succeq0$ and is defined by
\begin{equation}\label{eq:jrf}
\left\{
\begin{aligned}
 Y^{(k)}&\coloneqq\Bigl(\sum_i\sigma_iM_i^{(k)}\sigma_i\Bigr)^{1/2},\\
 M_i^{(k+1)}&\coloneqq \bigl(Y^{(k)}\bigr)^{-1}\sigma_iM_i^{(k)}\sigma_i\bigl(Y^{(k)}\bigr)^{-1},
\end{aligned}
\right.
\end{equation}
as long as $Y^{(k)}\succ0$. Its effects satisfy $M_i^{(k+1)}\succeq0$ and $\sum_iM_i^{(k+1)}=I_B$. The seed family need not be normalized. Seed matrices with $\ker M_i^{(0)}\cap\supp(\sigma_i)=\{0\}$, among them all positive definite ones, together with $\sum_i\sigma_i\succ0$ ensure that every step is defined, by Corollary~\ref{cor:jrf}. Convenient normalized identity seeds for the original algorithms are $X^{(0)}\coloneqq I_{AB}/d_A$ and $M_i^{(0)}\coloneqq I_B/m$.

\begin{theorem}[JRF as a block diagonal RW iteration]\label{thm:jrf-rw}
Let $C_{\mathrm{cq}}$ be the classical-to-quantum cost matrix in Eq.~\eqref{eq:cq-cost}, and let
\begin{equation}\label{eq:jrf-seed}
 X^{(0)}\coloneqq\sum_i\ket i\bra i\otimes M_i^{(0)}.
\end{equation}
Then the original RW iteration in Eq.~\eqref{eq:rw-psd} with cost matrix $C=C_{\mathrm{cq}}$ remains block diagonal,
\begin{equation}\label{eq:jrf-block-diagonal}
 X^{(k)}=\sum_i\ket i\bra i\otimes M_i^{(k)},
\end{equation}
and its blocks evolve exactly by the JRF iteration in Eq.~\eqref{eq:jrf}.
\end{theorem}
\begin{proof}
This correspondence also appears in the directional iteration framework~\cite[Sec.~3.1 and Appendix~B]{Tyson2010}. Direct block multiplication gives
\begin{equation}\label{eq:jrf-block}
 C_{\mathrm{cq}}X^{(k)}C_{\mathrm{cq}}
 =\sum_i\ket i\bra i\otimes \sigma_iM_i^{(k)}\sigma_i .
\end{equation}
Taking the partial trace over $A$ and then the positive square root gives
\begin{equation}\label{eq:jrf-normalizer}
 Y^{(k)}=\Bigl(\sum_i\sigma_iM_i^{(k)}\sigma_i\Bigr)^{1/2}.
\end{equation}
Substitution into Eq.~\eqref{eq:rw-psd} gives
\begin{equation}\label{eq:jrf-block-update}
 X^{(k+1)}=\sum_i\ket i\bra i\otimes\bigl(Y^{(k)}\bigr)^{-1}\sigma_iM_i^{(k)}\sigma_i\bigl(Y^{(k)}\bigr)^{-1}.
\end{equation}
This is exactly the block form of Eq.~\eqref{eq:jrf}.
\end{proof}

The JRF iteration is also related to the pretty good measurement (PGM)~\cite{HausladenWootters1994,Tyson2009} and to its generalizations~\cite{Mochon2006,TysonWeighted2009,AudenaertMosonyi2014,ZhouChessaChitambarLeditzky2025}. For a real exponent $\alpha$, the power PGM~\cite{AudenaertMosonyi2014,ZhouChessaChitambarLeditzky2025} has effects
\begin{equation}\label{eq:pgm}
 M_{i,\alpha}^{\mathrm{PGM}}\coloneqq\Bigl(\sum_j\sigma_j^{\alpha}\Bigr)^{-\frac{1}{2}}\sigma_i^{\alpha}\Bigl(\sum_j\sigma_j^{\alpha}\Bigr)^{-\frac{1}{2}},
\end{equation}
where a power of a singular $\sigma_i$ is taken on its support, and $\alpha=1$ gives the PGM. One JRF step from the seed $M_i^{(0)}\coloneqq\sigma_i^{\alpha-2}$, that is, $X^{(0)}=C_{\mathrm{cq}}^{\alpha-2}$ in the RW iteration, generates $M_i^{(1)}=M_{i,\alpha}^{\mathrm{PGM}}$, because $\sigma_iM_i^{(0)}\sigma_i=\sigma_i^{\alpha}$. This seed satisfies the zero-shift support condition of Eq.~\eqref{eq:support-condition} for every $\alpha$, since $\ker\sigma_i^{\alpha-2}=\ker\sigma_i$ meets $\supp(\sigma_i)$ only in zero. For $\alpha=2$, the uniform initialization $M_i^{(0)}=I_B/m$ yields the same first iterate, which is the quadratically weighted measurement~\cite{Tyson2009,TysonWeighted2009}.

The complementary cost matrix $C_{\mathrm{excl}}$ in Eq.~\eqref{eq:exclusion-cost} has diagonal blocks $\tau_i\coloneqq(\overline{\rho}-\sigma_i)/(m-1)$. The seed $X^{(0)}\coloneqq C_{\mathrm{excl}}^{-1}$, equivalently $M_i^{(0)}=\tau_i^{-1}$ in the corresponding JRF iteration, gives $Y^{(0)}=\overline{\rho}^{1/2}$ and $M_i^{(1)}=(I_B-M_{i,1}^{\mathrm{PGM}})/(m-1)$, which is the pretty bad measurement~\cite{McIrvinMohanSikora2024}. The inverse convention also covers singular $\tau_i$, since $\ker\tau_i^{-1}=\ker\tau_i$ ensures the required support condition.

With standard dense matrix arithmetic, each generalized RW iteration requires $\OCal(d_A^3d_B^3)$ arithmetic operations and storage for $\OCal(d_A^2d_B^2+d_B^2)$ scalar entries when $X^{(k)}$ and $S^{(k)}$ are stored separately. At zero shift, the original RW iteration omits the slack matrix, retaining the same arithmetic order and requiring $\OCal(d_A^2d_B^2)$ storage. The JRF iteration handles the $m=d_A$ effects separately, requiring $\OCal(d_A d_B^3)$ arithmetic operations and $\OCal(d_A d_B^2)$ storage per step.

These operation counts reflect the structure exploited by the original RW and JRF iterations. Their steps preserve positivity by congruence and restore normalization through a $d_B\times d_B$ matrix, requiring matrix products and a spectral decomposition of the normalizer. This avoids the semidefinite cone projections used by SCS~\cite{ODonoghue2016} and the search direction linear systems solved by interior point methods such as Clarabel~\cite{GoulartChen2026}. Direct implementation also avoids constructing a general conic model. Under the timing conventions of Appendix~\ref{sec:runtime-details}, the tested RW and JRF implementations have lower median runtimes than the CVXPY implementations using these solvers at every shared dimension in Figure~\ref{fig:timing}. The comparison uses a common objective accuracy target assessed against a precomputed dual reference. Since total work also depends on the iteration number, these finite comparisons do not establish an asymptotic exponent. The following section establishes convergence before Section~\ref{sec:complexity} derives iteration bounds.

\section{Global Convergence}
\label{sec:convergence}

We prove convergence of the generalized RW iteration by adjoining the slack matrix to the original variable. Throughout this section, $C=C^\dagger$, the shift $s\ge0$ satisfies Eq.~\eqref{eq:admissible-shift}, and the positive semidefinite seed pair satisfies Eq.~\eqref{eq:support-condition}. We represent this augmented formulation by
\begin{subequations}\label{eq:augmented-data}
\begin{align}
 \HCal_{A_+}&\coloneqq\HCal_A\oplus\Cbb,
 \qquad d_{A_+}\coloneqq d_A+1,\label{eq:augmented-space}\\
 C_+&\coloneqq\widetilde C\oplus sI_B,
 \qquad X_+^{(k)}\coloneqq X^{(k)}\oplus S^{(k)}.\label{eq:augmented-matrices}
\end{align}
\end{subequations}
These matrices act on $\HCal_{A_+}\otimes\HCal_B\cong(\HCal_A\otimes\HCal_B)\oplus\HCal_B$. The shift and support conditions give
\begin{equation}\label{eq:augmented-hypotheses}
\begin{split}
 C_+&\succeq0,\qquad
 \Omega_+\coloneqq\tr_{A_+}C_+\succ0,\\
 \ker X_+^{(0)}&\cap\supp(C_+)=\{0\}.
\end{split}
\end{equation}
Indeed, $\Omega_+=\tr_A\widetilde C+sI_B$, and the last condition separates into the two support conditions in Eq.~\eqref{eq:support-condition}. The proof first establishes convergence on the support of $C_+$, then excludes a singular limiting normalizer, and finally constructs a dual certificate for the original Hermitian problem.

\subsection{Support and polar normalization}

Assuming that the required inverse exists, the recursive Gram factors and their relation to the iterates are given by
\begin{subequations}\label{eq:factor-update}
\begin{align}
 G^{(0)}&\coloneqq \bigl(X_+^{(0)}\bigr)^{1/2},\label{eq:factor-init}\\
 G^{(k+1)}&\coloneqq G^{(k)}C_+(I_{A_+}\otimes \bigl(Y^{(k)}\bigr)^{-1}),\label{eq:factor-recursion}\\
 X_+^{(k)}&= \bigl(G^{(k)}\bigr)^\dagger G^{(k)}.\label{eq:factor-product}
\end{align}
\end{subequations}
By Eq.~\eqref{eq:factor-product}, $X_+^{(k)}$ is the Gram matrix of the columns of $G^{(k)}$, and we accordingly refer to $G^{(k)}$ as a Gram factor of $X_+^{(k)}$.
Let $\KCal\coloneqq\supp(C_+)\subseteq\HCal_{A_+}\otimes\HCal_B$, and let $\Pi_{\KCal}$ be its orthogonal projector.
Define
\begin{subequations}\label{eq:compressed}
\begin{align}
 \Gamma^{(k)}&\coloneqq C_+^{1/2}\bigl(G^{(k)}\bigr)^\dagger,\label{eq:compressed-w}\\
 H^{(k)}&\coloneqq\Pi_{\KCal} C_+^{1/2}(I_{A_+}\otimes \bigl(Y^{(k)}\bigr)^{-1})C_+^{1/2}\Pi_{\KCal}.\label{eq:compressed-h}
\end{align}
\end{subequations}
Here $\Gamma^{(k)}\colon\Cbb^{d_{A_+} d_B}\to \KCal$, and we regard the displayed projector sandwich as an operator $H^{(k)}\colon\KCal\to\KCal$. Positivity and invertibility of $H^{(k)}$ always refer to this support space. Its extension to the full tensor product vanishes on $\KCal^\perp$.
The following induction justifies all ordinary inverses in these formulas.

The last condition in Eq.~\eqref{eq:augmented-hypotheses} makes $\Gamma^{(0)}$ onto $\KCal$. Indeed, $\Gamma^{(0)}$ is onto $\KCal$ if and only if its adjoint $G^{(0)}C_+^{1/2}$ is injective on $\KCal$, that is, if and only if $C_+^{1/2}\ket{v}\notin\ker G^{(0)}=\ker X_+^{(0)}$ for every nonzero $\ket{v}\in\KCal$. Since $C_+^{1/2}$ maps $\KCal$ bijectively onto itself, this is precisely the support condition in Eq.~\eqref{eq:augmented-hypotheses}, which is therefore also equivalent to $\Gamma^{(0)}\bigl(\Gamma^{(0)}\bigr)^\dagger=C_+^{1/2}X_+^{(0)}C_+^{1/2}\succ0$ on $\KCal$.
Suppose $\Gamma^{(k)}$ is onto and all preceding factors are defined. Then $C_+X_+^{(k)}C_+=C_+^{1/2}\Gamma^{(k)}\bigl(\Gamma^{(k)}\bigr)^\dagger C_+^{1/2}$ has kernel $\ker C_+$. Because $\ker Y^{(k)}=\ker\bigl(Y^{(k)}\bigr)^2$, positivity implies that $\ket{b}\in\ker Y^{(k)}$ can hold only if $\HCal_{A_+}\otimes\operatorname{span}\{\ket{b}\}\subseteq\ker(C_+X_+^{(k)}C_+)=\ker C_+$. Taking the partial trace gives $\ket{b}\in\ker\Omega_+=\{0\}$. Hence $Y^{(k)}\succ0$, and $H^{(k)}\succ0$ on $\KCal$.
The next factors are therefore defined, and
$\Gamma^{(k+1)}=H^{(k)}\Gamma^{(k)}$ remains onto.
This induction proves invertibility of every normalizer. Block multiplication in Eq.~\eqref{eq:factor-recursion} preserves $X_+^{(k)}=X^{(k)}\oplus S^{(k)}$ and gives exactly Eq.~\eqref{eq:rw}. In particular,
\begin{equation}\label{eq:augmented-update}
\begin{split}
 Y^{(k)}&=\bigl[\tr_{A_+}(C_+X_+^{(k)}C_+)\bigr]^{1/2},\\
 X_+^{(k+1)}&=(I_{A_+}\otimes Y^{(k)})^{-1}C_+X_+^{(k)}C_+\\
 &\qquad\times(I_{A_+}\otimes Y^{(k)})^{-1}.
\end{split}
\end{equation}
It follows that $\tr_{A_+}X_+^{(k)}=I_B$ for $k\ge1$, which is the pair normalization in Eq.~\eqref{eq:pair-normalization}.

The reshaping $\hat{\Theta}$ is defined by the following Penrose diagram and equivalent algebraic expression.
\begin{equation}
\begin{aligned}
 \hat{\Theta}(G)&=\hat{\Theta}\left(\,
 \begin{tikzpicture}[penrose,baseline={([yshift=-.5ex]G.center)}]
  \node[tensor] (G) {$G$};
  \draw ([yshift=2.1mm]G.west) -- ++(-3mm,0) node[left] {$A_+$};
  \draw ([yshift=-2.1mm]G.west) -- ++(-3mm,0) node[left] {$B$};
  \draw ([yshift=2.1mm]G.east) -- ++(3mm,0) node[right] {$A_+'$};
  \draw ([yshift=-2.1mm]G.east) -- ++(3mm,0) node[right] {$B'$};
 \end{tikzpicture}
 \,\right)\\
 &\coloneqq\;
 \begin{tikzpicture}[penrose,baseline={([yshift=-.5ex]G.center)}]
  \node[tensor] (G) {$G$};
  \draw ([yshift=2.1mm]G.west) -- ++(-3mm,0) node[left] {$A_+$};
  \draw ([yshift=-2.1mm]G.west) -- ++(-3mm,0) node[left] {$B$};
  \draw ([yshift=-2.1mm]G.east) -- ++(3mm,0) node[right] {$B'$};
  \draw[rounded corners=2.4mm] ([yshift=2.1mm]G.east) -- ++(3.5mm,0)
   |- ([xshift=-3mm,yshift=3.6mm]G.north west) node[left] {$A_+''$};
 \end{tikzpicture}\\
 &=(I_{A_+''}\otimes G)\bigl(\dket{I}_{A_+''A_+'}\otimes I_{B'}\bigr).
\end{aligned}
\label{eq:reshaping}
\end{equation}
Here $\dket{I}_{A_+''A_+'}\coloneqq\sum_{a=1}^{d_{A_+}}\ket a\otimes\ket a$ is the unnormalized maximally entangled vector on two copies of $\HCal_{A_+}$.
The matrix $V\coloneqq\hat{\Theta}(G)$ has size $(d_{A_+}^2d_B)\times d_B$ and entries
\begin{equation}\label{eq:reshaping-entries}
 V_{(a',a,b),b'}=G_{(a,b),(a',b')}.
\end{equation}
The reshaping $\hat{\Theta}$ of $G$ gives
\begin{equation}
 V^\dagger V=\tr_{A_+}(G^\dagger G).
 \label{eq:reshaping-normalization}
\end{equation}
Consequently, for every $k\ge1$, the matrix $V^{(k)}\coloneqq\hat{\Theta}(G^{(k)})$ is an isometry and thus belongs to the complex \emph{Stiefel manifold}
\begin{equation}
 \MCal\coloneqq
 \{V\in\Cbb^{(d_{A_+}^2d_B)\times d_B}:V^\dagger V=I_B\}.
 \label{eq:stiefel}
\end{equation}
The proof uses this full manifold, even though the Gram factors of the iteration remain block diagonal. Although the entries of $V$ are complex, the constraint $V^\dagger V=I_B$ is not holomorphic, and we therefore regard $\MCal$ as a compact real analytic submanifold of $\Cbb^{(d_{A_+}^2d_B)\times d_B}\cong\Rbb^{2d_{A_+}^2d_B^2}$ of real dimension $2d_{A_+}^2d_B^2-d_B^2$. We equip $\MCal$ with the Riemannian metric induced by the real part of the Frobenius inner product, $\operatorname{Re}\finner{\cdot}{\cdot}$. The tangent space follows by differentiating the constraint, and the normal space is its orthogonal complement for this metric. Explicitly,
\begin{subequations}\label{eq:tangent-normal}
\begin{align}
 T_V\MCal&=\{U:U^\dagger V+V^\dagger U=0\},\label{eq:tangent-space}\\
 N_V\MCal&=\{VW:W=W^\dagger\}.\label{eq:normal-space}
\end{align}
\end{subequations}
Here $U$ has the size of $V$ and $W$ acts on $\HCal_B$. Indeed, $V^\dagger U$ is skew-Hermitian for every $U\in T_V\MCal$, so that $\finner{VW}{U}=\tr(W V^\dagger U)$ is purely imaginary, while the Hermitian matrices $W$ supply exactly the $d_B^2$ real dimensions missing from $T_V\MCal$.

The cost matrix enters the iteration through right multiplication of $G$ by $C_+$. The reshaping transports this action to the linear map $\hat{\Xi}_{C_+}$ on $\Cbb^{(d_{A_+}^2d_B)\times d_B}$, defined by
\begin{equation}\label{eq:transported-cost}
\begin{aligned}
 \hat{\Xi}_{C_+}(V)&=\hat{\Xi}_{C_+}\left(\,
 \begin{tikzpicture}[penrose,baseline={([yshift=-.5ex]V.center)}]
  \node[tensor,minimum width=8mm,minimum height=12mm] (V) {$V$};
  \draw ([yshift=4mm]V.west) -- ++(-3mm,0) node[left] {$A_+''$};
  \draw (V.west) -- ++(-3mm,0) node[left] {$A_+$};
  \draw ([yshift=-4mm]V.west) -- ++(-3mm,0) node[left] {$B$};
  \draw ([yshift=-4mm]V.east) -- ++(3mm,0) node[right] {$B'$};
 \end{tikzpicture}
 \,\right)\\
 &\coloneqq\;
 \begin{tikzpicture}[penrose,baseline={([yshift=-.5ex]V.center)}]
  \node[tensor,minimum width=8mm,minimum height=12mm] (V) at (0,0) {$V$};
  \node[tensor] (C) at (14mm,-2mm) {$C_+$};
  \draw (V.west) -- ++(-6mm,0) node[left] {$A_+$};
  \draw ([yshift=-4mm]V.west) -- ++(-6mm,0) node[left] {$B$};
  \draw ([yshift=-4mm]V.east) -- ([yshift=-2mm]C.west);
  \draw[rounded corners=1.5mm] ([yshift=4mm]V.west) -- ++(-3mm,0) -- ++(0,4mm)
   -- ++(14mm,0) -- ++(0,-8mm) -- ([yshift=2mm]C.west);
  \draw ([yshift=-2mm]C.east) -- ++(3mm,0) node[right] {$B'$};
  \draw[rounded corners=2mm] ([yshift=2mm]C.east) -- ++(3mm,0) -- ++(0,10.5mm)
   -- ([xshift=-6mm,yshift=4.5mm]V.north west) node[left] {$A_+''$};
  \path ([yshift=1.5mm]current bounding box.north);
 \end{tikzpicture}\\
 &=\hat{\Theta}\bigl(\hat{\Theta}^{-1}(V)C_+\bigr).
\end{aligned}
\end{equation}
The map $\hat{\Theta}$ is a bijective Frobenius isometry, so its inverse $\hat{\Theta}^{-1}$ equals its adjoint. Because $\hat{\Theta}$ is an isometry and $\finner{G}{G'C_+}=\tr(C_+G^\dagger G')$, the map $\hat{\Xi}_{C_+}$ is self-adjoint and positive semidefinite for the Frobenius inner product. Moreover, since $C_+^2\preceq\opnorm{C_+} C_+$ and $\hat{\Xi}_{C_+}^2(V)=\hat{\Theta}\bigl(\hat{\Theta}^{-1}(V)C_+^2\bigr)$, we have
\begin{equation}\label{eq:transported-cost-bound}
 \finner{V}{\hat{\Xi}_{C_+}^2(V)}\le\opnorm{C_+}\finner{V}{\hat{\Xi}_{C_+}(V)}
\end{equation}
for every $V$. The augmented objective on $\MCal$ is the real analytic function
\begin{equation}\label{eq:objective-manifold}
\begin{aligned}
 f_+(V)&\coloneqq\finner{V}{\hat{\Xi}_{C_+}(V)}\\
 &=\;
 \begin{tikzpicture}[penrose,baseline={([yshift=-.5ex]V.center)}]
  \node[tensor,minimum width=8mm,minimum height=12mm] (Vd) at (-15.2mm,0) {$V^\dagger$};
  \node[tensor,minimum width=8mm,minimum height=12mm] (V) at (0,0) {$V$};
  \node[tensor] (C) at (14mm,-2mm) {$C_+$};
  \draw (Vd.east) -- (V.west);
  \draw ([yshift=-4mm]Vd.east) -- ([yshift=-4mm]V.west);
  \draw ([yshift=-4mm]V.east) -- ([yshift=-2mm]C.west);
  \draw[rounded corners=1.5mm] ([yshift=4mm]V.west) -- ++(-3mm,0) -- ++(0,4mm)
   -- ++(14mm,0) -- ++(0,-8mm) -- ([yshift=2mm]C.west);
  \draw[rounded corners=1.5mm] ([yshift=2mm]C.east) -- (21mm,0) -- (21mm,10.5mm)
   -- (-8.74727mm,10.5mm) -- (-8.74727mm,4mm) -- ([yshift=4mm]Vd.east);
  \draw[rounded corners=1.5mm] ([yshift=-2mm]C.east) -- (21mm,-4mm) -- (21mm,-8.8mm)
   -- (-22.2mm,-8.8mm) -- (-22.2mm,-4mm) -- ([yshift=-4mm]Vd.west);
 \end{tikzpicture}.
\end{aligned}
\end{equation}
Along the iterates,
\begin{equation}\label{eq:objective-iterates}
\begin{aligned}
 f_+^{(k)}\coloneqq f_+(V^{(k)})
 &=\finner{G^{(k)}}{G^{(k)}C_+}\\
 &=\tr(C_+X_+^{(k)})=\fnorm{\Gamma^{(k)}}^2,
\end{aligned}
\end{equation}
where the cost matrix stands to the right of $G^{(k)}$ because $X_+^{(k)}=(G^{(k)})^\dagger G^{(k)}$. We retain $f^{(k)}\coloneqq\tr(CX^{(k)})$ for the original objective. At every normalized index,
\begin{equation}\label{eq:objective-offset}
\begin{split}
 f_+^{(k)}
 &=\tr(CX^{(k)})
   +s\tr\bigl(\tr_AX^{(k)}+S^{(k)}\bigr)\\
 &=f^{(k)}+s\,d_B.
\end{split}
\end{equation}
Thus $f^{(k)}$ can be negative even though $f_+^{(k)}\ge0$. Since $\hat{\Theta}(G(I_{A_+}\otimes T))=\hat{\Theta}(G)T$ for every operator $T$ on $\HCal_B$, Eq.~\eqref{eq:factor-recursion} equivalently gives
\begin{equation}
 \hat{\Xi}_{C_+}(V^{(k)})=V^{(k+1)}Y^{(k)}.
 \label{eq:polar-identity}
\end{equation}
Thus $V^{(k+1)}$ is the unique column polar factor of $\hat{\Xi}_{C_+}(V^{(k)})$.

For $k\ge0$, define the matrix increments
\begin{equation}\label{eq:matrix-increments}
\begin{split}
 \delta V^{(k)}&\coloneqq V^{(k+1)}-V^{(k)},\\
 \delta \Gamma^{(k)}&\coloneqq \Gamma^{(k+1)}-\Gamma^{(k)}.
\end{split}
\end{equation}

\begin{lemma}[Monotonicity]
\label{lem:ascent}
Under an admissible shift $s$ (Eq.~\eqref{eq:admissible-shift}) and a seed pair satisfying Eq.~\eqref{eq:support-condition}, the following bounds hold for every $k\ge1$.
\begin{equation}
\begin{split}
 f_+^{(k)}&\le\tr Y^{(k)}\le f_+^{(k+1)},\\
 f_+^{(k+1)}-f_+^{(k)}&\ge\fnorm{\delta \Gamma^{(k)}}^2.
\end{split}
 \label{eq:ascent-bounds}
\end{equation}
For $\tr_{A_+} X_+^{(0)}=I_B$, they also hold at $k=0$. At the same indices,
\begin{equation}\label{eq:original-ascent}
 f^{(k)}\le\tr Y^{(k)}-s\,d_B\le f^{(k+1)}.
\end{equation}
\end{lemma}
\begin{proof}
At the stated indices, Eqs.~\eqref{eq:reshaping-normalization} and~\eqref{eq:factor-product} combine with the normalization $\tr_{A_+}X_+^{(k)}=I_B$ to give $(V^{(k)})^\dagger V^{(k)}=\tr_{A_+}X_+^{(k)}=I_B$, and likewise for $V^{(k+1)}$. Both matrices are therefore isometries from $\HCal_B$ into $\Cbb^{d_{A_+}^2d_B}$, which is precisely the statement that their columns are orthonormal. Writing $V^{(k+1)}=V^{(k)}+\delta V^{(k)}$ and expanding the quadratic objective, we obtain
\begin{subequations}\label{eq:ascent}
\begin{align}
 &\;f_+^{(k+1)}-f_+^{(k)}\nonumber\\
 &=2\operatorname{Re}\finner{\delta V^{(k)}}{\hat{\Xi}_{C_+}(V^{(k)})}+\fnorm{\delta \Gamma^{(k)}}^2\label{eq:ascent-expansion}\\
 &=2\operatorname{Re}\finner{\delta V^{(k)}}{V^{(k+1)}Y^{(k)}}+\fnorm{\delta \Gamma^{(k)}}^2\label{eq:ascent-polar}\\
 &=2\bigl(\tr Y^{(k)}-f_+^{(k)}\bigr)+\fnorm{\delta \Gamma^{(k)}}^2.\label{eq:ascent-gap}
\end{align}
\end{subequations}
Equation~\eqref{eq:ascent-expansion} expands $f_+(V^{(k)}+\delta V^{(k)})$ and evaluates the quadratic increment through Eq.~\eqref{eq:transported-cost} and the isometry of $\hat{\Theta}$, which give $\finner{\delta V^{(k)}}{\hat{\Xi}_{C_+}(\delta V^{(k)})}=\fnorm{C_+^{1/2}(G^{(k+1)}-G^{(k)})^\dagger}^2=\fnorm{\delta \Gamma^{(k)}}^2$. Equation~\eqref{eq:ascent-polar} then rewrites the first summand using Eq.~\eqref{eq:polar-identity}. Equation~\eqref{eq:ascent-gap} rests on two identities. Since $V^{(k+1)}$ is an isometry, $(V^{(k+1)})^\dagger V^{(k+1)}=I_B$ and consequently $\tr((V^{(k+1)})^\dagger V^{(k+1)}Y^{(k)})=\tr Y^{(k)}$, while Eq.~\eqref{eq:polar-identity} gives the cross term identity
\begin{equation}\label{eq:cross-term}
\begin{split}
 &\operatorname{Re}\tr\bigl((V^{(k)})^\dagger V^{(k+1)}Y^{(k)}\bigr)\\
 &=\finner{V^{(k)}}{\hat{\Xi}_{C_+}(V^{(k)})}=f_+^{(k)}.
\end{split}
\end{equation}
Subtracting the latter from the former yields $\operatorname{Re}\finner{\delta V^{(k)}}{V^{(k+1)}Y^{(k)}}=\tr Y^{(k)}-f_+^{(k)}$.

Using the isometry identities to expand the squared norm and Eq.~\eqref{eq:cross-term} to combine the two cross terms, we obtain
\begin{equation}\label{eq:gap-trace}
\begin{aligned}
 0&\le\fnorm{\delta V^{(k)}(Y^{(k)})^{1/2}}^2\\
  &=\tr\bigl(Y^{(k)}(\delta V^{(k)})^\dagger\delta V^{(k)}\bigr)\\
  &=2\tr Y^{(k)}-2\operatorname{Re}\tr\bigl((V^{(k)})^\dagger V^{(k+1)}Y^{(k)}\bigr)\\
  &=2\bigl(\tr Y^{(k)}-f_+^{(k)}\bigr).
\end{aligned}
\end{equation}
Hence $f_+^{(k)}\le\tr Y^{(k)}$. Inserting this bound into Eq.~\eqref{eq:ascent-gap} yields $f_+^{(k+1)}-f_+^{(k)}\ge\fnorm{\delta \Gamma^{(k)}}^2$. Rearranging the same equation as $f_+^{(k+1)}-\tr Y^{(k)}=\fnorm{\delta \Gamma^{(k)}}^2+\bigl(\tr Y^{(k)}-f_+^{(k)}\bigr)$ shows, in addition, that $\tr Y^{(k)}\le f_+^{(k+1)}$.

Finally, the objective values are bounded. For $k\ge1$, the normalization $\tr_{A_+}X_+^{(k)}=I_B$ gives $\tr X_+^{(k)}=d_B$, so that Eq.~\eqref{eq:objective-iterates} and $X_+^{(k)}\succeq0$ yield $f_+^{(k)}=\tr(C_+X_+^{(k)})\le\opnorm{C_+}\tr X_+^{(k)}=d_B\opnorm{C_+}$. Being nondecreasing by Eq.~\eqref{eq:ascent-bounds} and bounded above, the sequence has a limit $f_+^{(\infty)}\coloneqq\lim_{k\to\infty}f_+^{(k)}$. Equation~\eqref{eq:objective-offset} then gives the original objective limit $f^{(\infty)}\coloneqq f_+^{(\infty)}-s\,d_B$ and the ascent bounds in Eq.~\eqref{eq:original-ascent}.
\end{proof}

Write $\nabla_{\MCal}f_+(V)$ for the Riemannian gradient of $f_+$, namely the orthogonal projection onto $T_V\MCal$ of the Euclidean gradient, which equals $2\hat{\Xi}_{C_+}(V)$ because $\hat{\Xi}_{C_+}$ is self-adjoint.

\begin{lemma}
\label{lem:gradient-bound}
Under an admissible shift $s$ (Eq.~\eqref{eq:admissible-shift}) and a seed pair satisfying Eq.~\eqref{eq:support-condition}, let $\kappa\coloneqq 2\opnorm{C_+}^{1/2}$. Then the following bound holds for every $k\ge1$.
\begin{equation}
 \fnorm{\nabla_{\MCal}f_+(V^{(k+1)})}\le \kappa\fnorm{\delta \Gamma^{(k)}}.
 \label{eq:gradient-bound}
\end{equation}
\end{lemma}
\begin{proof}
Let $P^{(k+1)}$ denote the orthogonal projection onto the tangent space $T_{V^{(k+1)}}\MCal$ at $V^{(k+1)}\in\MCal$. Since $Y^{(k)}$ is Hermitian, Eq.~\eqref{eq:polar-identity} places $\hat{\Xi}_{C_+}(V^{(k)})$ in the normal space of Eq.~\eqref{eq:normal-space} at $V^{(k+1)}$, where $P^{(k+1)}$ annihilates it. Linearity and the identity $V^{(k+1)}=V^{(k)}+\delta V^{(k)}$ then yield
\begin{equation}\label{eq:gradient-step-identity}
 \nabla_{\MCal}f_+(V^{(k+1)})=2P^{(k+1)}\bigl(\hat{\Xi}_{C_+}(\delta V^{(k)})\bigr).
\end{equation}
Projection does not increase the norm, so Eq.~\eqref{eq:transported-cost-bound}, together with the identity $\finner{\delta V^{(k)}}{\hat{\Xi}_{C_+}(\delta V^{(k)})}=\fnorm{\delta \Gamma^{(k)}}^2$ established in the proof of Lemma~\ref{lem:ascent}, gives
\begin{equation}\label{eq:gradient-step-estimate}
\begin{aligned}
 \fnorm{\nabla_{\MCal}f_+(V^{(k+1)})}^2
 &\le4\opnorm{C_+}\finner{\delta V^{(k)}}{\hat{\Xi}_{C_+}(\delta V^{(k)})}\\
 &=4\opnorm{C_+}\fnorm{\delta \Gamma^{(k)}}^2.
\end{aligned}
\end{equation}
\end{proof}

\subsection{Convergence before taking an inverse limit}

We now show that the transformed iterates $\Gamma^{(k)}$ converge. This suffices for the normalizers, because Eqs.~\eqref{eq:rw}, \eqref{eq:factor-product}, and~\eqref{eq:compressed-w} give
\begin{equation}\label{eq:normalizer-finite}
 \bigl(Y^{(k)}\bigr)^2=\tr_{A_+}(C_+^{1/2}\Gamma^{(k)}\bigl(\Gamma^{(k)}\bigr)^\dagger C_+^{1/2}),
\end{equation}
which, by continuity of the positive square root on PSD matrices, expresses $Y^{(k)}$ as a continuous function of $\Gamma^{(k)}$ without any inverse. Lemma~\ref{lem:nonsingular} then establishes that the limiting normalizer is nonsingular, and only thereafter does an inverse enter the limit.

\begin{lemma}[Convergence of the normalizer]
\label{lem:finite-length}
Under an admissible shift $s$ (Eq.~\eqref{eq:admissible-shift}) and a seed pair satisfying Eq.~\eqref{eq:support-condition}, $\sum_k\fnorm{\delta \Gamma^{(k)}}<\infty$, so that $\Gamma^{(k)}$ converges to a limit $\Gamma^{(\infty)}$. Consequently, with
\begin{equation}\label{eq:normalizer-convergence}
 Y^{(\infty)}\coloneqq\Bigl[\tr_{A_+}(C_+^{1/2}\Gamma^{(\infty)}\bigl(\Gamma^{(\infty)}\bigr)^\dagger C_+^{1/2})\Bigr]^{1/2},
\end{equation}
the normalizers satisfy
\begin{equation}\label{eq:normalizer-limits}
 \lim_{k\to\infty}\opnorm{Y^{(k)}-Y^{(\infty)}}=0.
\end{equation}
\end{lemma}
\begin{proof}
A point $V_\star\in\MCal$ is \emph{critical} for $f_+$ if its derivative vanishes in every tangent direction, equivalently if $\nabla_{\MCal}f_+(V_\star)=0$. The required analyticity follows from the polynomial constraint $\Psi(V)\coloneqq V^\dagger V-I_B$. For an arbitrary matrix perturbation $U$ of the same size as $V$ and a real parameter $t$, its derivative is $\left.\frac{\dd}{\dd t}\Psi(V+tU)\right|_{t=0}=U^\dagger V+V^\dagger U$. At each $V\in\MCal$, this derivative attains every Hermitian matrix $W\in\lin(\HCal_B)$ by choosing $U=VW/2$, since $V^\dagger V=I_B$. The real analytic implicit function theorem therefore supplies analytic charts on $\MCal$. The objective in Eq.~\eqref{eq:objective-manifold}, being quadratic in the real and imaginary parts of $V$, is real analytic in these charts.

At every critical point $V_\star$, the \L{}ojasiewicz gradient inequality~\cite{HarauxJendoubi2015} provides constants $\xi>0$ and $\theta\in(0,1)$ such that
\begin{equation}
 |f_+(V)-f_+(V_\star)|^\theta\le\xi\fnorm{\nabla_{\MCal}f_+(V)}
 \label{eq:lojasiewicz}
\end{equation}
for all $V\in\MCal$ near $V_\star$. The Euclidean result~\cite[Lemma~2.1]{AbsilMahonyAndrews2005} extends to $\MCal$ because coordinate and Riemannian gradient norms are locally comparable.

Unless $f_+$ is locally constant, its first nonconstant Taylor term at $V_\star$ in these coordinates has degree $p\ge2$ by criticality. Along a direction where this term is nonzero, the objective difference and gradient norm vanish to orders $p$ and $p-1$, respectively. Equation~\eqref{eq:lojasiewicz} thus requires $p\theta\ge p-1$, giving $\theta\ge1-1/p\ge1/2$. The locally constant case permits $\theta=1/2$, so we may choose $\theta\in[1/2,1)$. By continuity, we also shrink the neighborhood so that $|f_+(V)-f_+(V_\star)|\le1$.

Since convergence of $V^{(k)}$ is not yet known, we work with its accumulation points, namely the limits of convergent subsequences. Their set $\ACal\subseteq\MCal$ is nonempty and compact because $\MCal$ is compact.

Every $V_\star\in\ACal$ is critical with $f_+(V_\star)=f_+^{(\infty)}$, as two separate continuity arguments show along a subsequence $V^{(k_j)}\to V_\star$. First, continuity of $f_+$ gives $f_+(V_\star)=\lim_{j\to\infty}f_+^{(k_j)}=f_+^{(\infty)}$, because the entire sequence $f_+^{(k)}$ converges to $f_+^{(\infty)}$. Second, Eq.~\eqref{eq:ascent-bounds} gives $\fnorm{\delta \Gamma^{(k)}}^2\le f_+^{(k+1)}-f_+^{(k)}\to0$, whence Lemma~\ref{lem:gradient-bound} implies $\nabla_{\MCal}f_+(V^{(k)})\to0$, so that continuity of $\nabla_{\MCal}f_+$ yields $\nabla_{\MCal}f_+(V_\star)=0$. We cover $\ACal$ by finitely many neighborhoods on which Eq.~\eqref{eq:lojasiewicz} holds, and henceforth let $\xi$ and $\theta$ denote the largest of the associated constants and exponents. Because raising the exponent does not increase a quantity bounded by one, Eq.~\eqref{eq:lojasiewicz} then holds with this single pair, and with $f_+(V_\star)$ replaced by $f_+^{(\infty)}$, at every point $V$ of the union of these neighborhoods that satisfies $|f_+(V)-f_+^{(\infty)}|\le1$. Moreover, $V^{(k)}$ lies in this union for every sufficiently large $k$, since a subsequence remaining outside the open union would, by compactness, accumulate at a point outside it, contrary to the definition of $\ACal$.

Define the residual to the limiting objective value
\begin{equation}\label{eq:gap}
 r^{(k)}\coloneqq f_+^{(\infty)}-f_+^{(k)},\qquad k\ge1,
\end{equation}
which is nonnegative, nonincreasing by Lemma~\ref{lem:ascent}, and convergent to zero. At this stage, the limiting value has not yet been identified with the optimum. A degenerate case can be set aside first. If $\delta \Gamma^{(k)}=0$ for some $k\ge1$, then $f_+^{(k+1)}=f_+^{(k)}$, whereupon Eqs.~\eqref{eq:ascent-gap} and~\eqref{eq:gap-trace} together with $Y^{(k)}\succ0$ give $\delta V^{(k)}=0$. The iteration is then stationary from step $k$ onward, and the lemma holds trivially. For the remainder of the proof we therefore assume that $\delta \Gamma^{(k)}\neq0$ for all $k\ge1$, which also forces $r^{(k)}>0$, because Eq.~\eqref{eq:ascent-bounds} gives
\begin{equation}\label{eq:gap-decrement}
 r^{(k)}-r^{(k+1)}=f_+^{(k+1)}-f_+^{(k)}\ge\fnorm{\delta \Gamma^{(k)}}^2.
\end{equation}
Define the comparison function $\varphi(t)\coloneqq\xi t^{1-\theta}/(1-\theta)$ for $t>0$, which is concave and has derivative $\varphi'(t)=\xi t^{-\theta}$. For every sufficiently large $k$, Eq.~\eqref{eq:lojasiewicz} applied at $V=V^{(k)}$ reads $(r^{(k)})^\theta\le\xi\fnorm{\nabla_{\MCal}f_+(V^{(k)})}$, which is equivalent to
\begin{equation}\label{eq:lojasiewicz-derivative}
 \varphi'(r^{(k)})\fnorm{\nabla_{\MCal}f_+(V^{(k)})}\ge1.
\end{equation}
Since Eq.~\eqref{eq:gradient-bound} with index $k-1$ bounds this gradient by $\kappa\fnorm{\delta \Gamma^{(k-1)}}$, we infer that $\varphi'(r^{(k)})\ge1/(\kappa\fnorm{\delta \Gamma^{(k-1)}})$. Concavity of $\varphi$, on the other hand, gives the tangent line inequality
\begin{equation}\label{eq:tangent-line}
\varphi(r^{(k)})-\varphi(r^{(k+1)})\ge
\varphi'(r^{(k)})(r^{(k)}-r^{(k+1)}).
\end{equation}
Combining these two bounds with Eq.~\eqref{eq:gap-decrement}, we obtain
\begin{equation}\label{eq:comparison-decrement}
\begin{split}
 &\varphi(r^{(k)})-\varphi(r^{(k+1)})\ge\frac{\fnorm{\delta \Gamma^{(k)}}^2}{\kappa\fnorm{\delta \Gamma^{(k-1)}}}\\
 &\qquad\ge\frac{2\fnorm{\delta \Gamma^{(k)}}-\fnorm{\delta \Gamma^{(k-1)}}}{\kappa},
\end{split}
\end{equation}
where the last step is the elementary inequality $a^2/b\ge2a-b$ for $a\ge0$ and $b>0$, a rearrangement of $(a-b)^2\ge0$. We now sum Eq.~\eqref{eq:comparison-decrement} over $k=j,\dots,j_{\max}$ for a sufficiently late $j$. The left side telescopes to at most $\varphi(r^{(j)})$, whereas on the right every increment with $j\le k\le j_{\max}-1$ retains the net coefficient $1/\kappa$. Discarding the nonnegative remainder and letting $j_{\max}\to\infty$ yields
\begin{equation}
 \sum_{k=j}^{\infty}\fnorm{\delta \Gamma^{(k)}}
 \le\fnorm{\delta \Gamma^{(j-1)}}+\kappa\varphi(r^{(j)}).
 \label{eq:finite-length}
\end{equation}
The increments are thus absolutely summable, so that $\Gamma^{(k)}$ is a Cauchy sequence and converges to a limit $\Gamma^{(\infty)}$. Finally, the right-hand side of Eq.~\eqref{eq:normalizer-finite} is polynomial in $\Gamma^{(k)}$ and its adjoint, and therefore converges. Continuity of the positive square root on PSD matrices then gives Eq.~\eqref{eq:normalizer-limits}, with the limit defined by Eq.~\eqref{eq:normalizer-convergence}.
\end{proof}

It remains to exclude a singular limit $Y^{(\infty)}$ in Lemma~\ref{lem:finite-length}. Define the auxiliary matrices $Q^{(0)}\coloneqq I_{\KCal}$ and $Q^{(k+1)}\coloneqq H^{(k)}Q^{(k)}$. The following lemma proves nonsingularity before taking an inverse limit and supplies the hypotheses needed for Lemma~\ref{lem:cone}.

\begin{lemma}[Nonsingular limiting normalizer]
\label{lem:nonsingular}
Under an admissible shift $s$ (Eq.~\eqref{eq:admissible-shift}) and a seed pair satisfying Eq.~\eqref{eq:support-condition}, the following limits and bounds hold,
\begin{subequations}\label{eq:nonsingular-conclusions}
\begin{align}
 Y^{(\infty)}&\succ0,\label{eq:z-positive}\\
 \lim_{k\to\infty}\opnorm{H^{(k)}-H^{(\infty)}}&=0,\label{eq:h-convergence}\\
 \sup_k\opnorm{Q^{(k)}}&<\infty,\label{eq:q-bounded}
\end{align}
\end{subequations}
and every $Q^{(k)}$ is invertible on $\KCal$.
\end{lemma}
\begin{proof}
We begin with two bounds that hold uniformly along the iteration and will be used to exclude a singular limiting normalizer. Let $\sigma_{\min}^{+}(C_+)$ be the smallest positive singular value of $C_+$ and put
$h\coloneqq \sigma_{\min}^{+}(C_+)/(\sqrt{d_B}\opnorm{C_+})$.
For $k\ge1$,
\begin{subequations}\label{eq:uniform-bounds}
\begin{align}
 \opnorm{Y^{(k)}}&\le \sqrt{d_B}\opnorm{C_+},\label{eq:uniform-bound-z}\\
 H^{(k)}&\succeq hI_{\KCal}.\label{eq:uniform-bound-h}
\end{align}
\end{subequations}
Equation~\eqref{eq:uniform-bound-z} follows from $\opnorm{Y^{(k)}}^2\le\tr\bigl((Y^{(k)})^2\bigr)=\tr(C_+^2X_+^{(k)})\le d_B\opnorm{C_+}^2$, in which $\tr X_+^{(k)}=d_B$ by normalization. Equation~\eqref{eq:uniform-bound-h} follows from Eq.~\eqref{eq:compressed-h}, because $(Y^{(k)})^{-1}\succeq\opnorm{Y^{(k)}}^{-1}I_B$ and $\Pi_{\KCal}C_+\Pi_{\KCal}\succeq\sigma_{\min}^{+}(C_+)I_{\KCal}$.

We first establish invertibility and boundedness of the auxiliary matrices on $\KCal$, whose dimension is denoted by $d_{\KCal}\coloneqq\dim\KCal<\infty$. For an operator $T$ on this space, $\det_{\KCal}T$ denotes the determinant of its matrix in an orthonormal basis. Since each $H^{(k)}$ is positive definite on $\KCal$, the recursion for $Q^{(k)}$ shows that $Q^{(k)}$ is invertible and $\det_{\KCal}Q^{(k)}>0$. To establish boundedness, we use the relation $\Gamma^{(k+1)}=H^{(k)}\Gamma^{(k)}$ from the preceding subsection, which gives $\Gamma^{(k)}=Q^{(k)}\Gamma^{(0)}$. Since $\Gamma^{(0)}$ is onto $\KCal$, its Moore--Penrose pseudoinverse, which maps $\KCal$ to $\Cbb^{d_{A_+} d_B}$, is
\begin{equation}\label{eq:pseudoinverse-w0}
 \bigl(\Gamma^{(0)}\bigr)^{-1}=\bigl(\Gamma^{(0)}\bigr)^\dagger
 \bigl[\Gamma^{(0)}\bigl(\Gamma^{(0)}\bigr)^\dagger\bigr]^{-1},
\end{equation}
where the inverse is taken on $\KCal$. Thus $\Gamma^{(0)}\bigl(\Gamma^{(0)}\bigr)^{-1}=I_{\KCal}$ and consequently $Q^{(k)}=\Gamma^{(k)}\bigl(\Gamma^{(0)}\bigr)^{-1}$, which is bounded because $\fnorm{\Gamma^{(k)}}^2=f_+^{(k)}\le d_B\opnorm{C_+}$ for $k\ge1$. This proves Eq.~\eqref{eq:q-bounded}.

Suppose now, for the sake of contradiction, that $Y^{(\infty)}$ were singular. Since the smallest eigenvalue of $Y^{(k)}$ would then tend to zero, and since $\Omega_+=\tr_{A_+}C_+\succ0$ has a positive smallest eigenvalue $\lambda_{\min}(\Omega_+)$, we would have
\begin{equation}\label{eq:compressed-trace-divergence}
\begin{split}
 \tr H^{(k)}
 &=\tr(\Omega_+\bigl(Y^{(k)}\bigr)^{-1})\\
 &\ge\lambda_{\min}(\Omega_+)\tr(\bigl(Y^{(k)}\bigr)^{-1})
 \xrightarrow{k\to\infty}+\infty.
\end{split}
\end{equation}
As every eigenvalue of $H^{(k)}$ is at least $h$ by Eq.~\eqref{eq:uniform-bound-h} while the largest one is at least their mean, the determinant would diverge as well,
\begin{equation}\label{eq:compressed-determinant-divergence}
 \det_{\KCal} H^{(k)}\ge\frac{h^{d_{\KCal}-1}}{d_{\KCal}}\tr H^{(k)}
 \xrightarrow{k\to\infty}+\infty.
\end{equation}
It would follow that
$\det_{\KCal} Q^{(k+1)}=\det_{\KCal}H^{(k)}\det_{\KCal}Q^{(k)}\ge2\det_{\KCal} Q^{(k)}>0$ for all sufficiently large $k$, and hence that $\det_{\KCal}Q^{(k)}\to\infty$, which is incompatible with the boundedness of $Q^{(k)}$. Consequently $Y^{(\infty)}\succ0$, which is Eq.~\eqref{eq:z-positive}. We may now define the operator on $\KCal$
\begin{equation}\label{eq:h-limit}
 H^{(\infty)}\coloneqq
 \Pi_{\KCal} C_+^{1/2}(I_{A_+}\otimes\bigl(Y^{(\infty)}\bigr)^{-1})C_+^{1/2}\Pi_{\KCal}.
\end{equation}
Because matrix inversion is continuous at every invertible matrix, $\bigl(Y^{(k)}\bigr)^{-1}$ converges to $\bigl(Y^{(\infty)}\bigr)^{-1}$, and Eq.~\eqref{eq:compressed-h} then yields Eq.~\eqref{eq:h-convergence}.
\end{proof}

Lemmas~\ref{lem:finite-length} and~\ref{lem:nonsingular} establish convergence of the transformed iterates and normalizers, together with nonsingularity of the limiting normalizer. It remains to recover convergence of the original matrices and to prove that $Y^{(\infty)}-sI_B$ is a dual certificate. The next lemma supplies the operator inequality needed for that certificate.

\subsection{Closing the duality gap}

Eigenvalues of a Hermitian operator $H$ are denoted by $\lambda$, and $\lambda_{\max}(H)$ denotes the largest of them. The following lemma bounds the limit $H^{(\infty)}$ by the identity, converting the boundedness of the auxiliary matrices into dual feasibility.

\begin{lemma}[Limiting spectral bound]
\label{lem:cone}
Let $H^{(k)}=\bigl(H^{(k)}\bigr)^\dagger\succ0$ converge to a limit $H^{(\infty)}$.
If an invertible matrix sequence $(Q^{(k)})$ satisfies
$Q^{(k+1)}=H^{(k)}Q^{(k)}$ and $\sup_k\opnorm{Q^{(k)}}<\infty$, then
$H^{(\infty)}\preceq I$.
\end{lemma}
\begin{proof}
Suppose instead that $\lambda\coloneqq\lambda_{\max}(H^{(\infty)})>1$. Fix $\beta\in(1,\lambda)$ that is not an eigenvalue of $H^{(\infty)}$, let $\eta>0$ be the distance from $\beta$ to the spectrum of $H^{(\infty)}$, and consider the quadratic form $\psi(\ket{x})\coloneqq\bra{x}(H^{(\infty)}-\beta I)\ket{x}$. As a limit of positive definite matrices, $H^{(\infty)}$ is positive semidefinite, so that $H^{(\infty)}+\beta I\succeq\beta I$. Since all functions of $H^{(\infty)}$ commute, we obtain
\begin{equation}\label{eq:form-expansion}
\begin{split}
 &\psi(H^{(\infty)}\ket{x})-\beta^2\psi(\ket{x})\\
 &=\bra{x}(H^{(\infty)}-\beta I)^2(H^{(\infty)}+\beta I)\ket{x}\\
 &\ge\beta\eta^2\norm{\ket{x}}_2^2
\end{split}
\end{equation}
for every $\ket{x}$. Set $\zeta_k\coloneqq\opnorm{H^{(k)}-H^{(\infty)}}$. Adding and subtracting $H^{(\infty)}(H^{(\infty)}-\beta I)H^{(k)}$ gives
\begin{equation}\label{eq:form-perturbation-expansion}
\begin{aligned}
 &H^{(k)}(H^{(\infty)}-\beta I)H^{(k)}\\
 &\quad-H^{(\infty)}(H^{(\infty)}-\beta I)H^{(\infty)}\\
 &=(H^{(k)}-H^{(\infty)})(H^{(\infty)}-\beta I)H^{(k)}\\
 &\quad+H^{(\infty)}(H^{(\infty)}-\beta I)(H^{(k)}-H^{(\infty)}).
\end{aligned}
\end{equation}
Cauchy--Schwarz and submultiplicativity of the operator norm then give
\begin{equation}\label{eq:form-perturbation-bound}
\begin{aligned}
 &\bigl|\psi(H^{(k)}\ket{x})-\psi(H^{(\infty)}\ket{x})\bigr|\\
 &\quad\le\opnorm{H^{(\infty)}-\beta I}
   \bigl(\opnorm{H^{(k)}}+\opnorm{H^{(\infty)}}\bigr)
   \zeta_k\norm{\ket{x}}_2^2\\
 &\quad\le\opnorm{H^{(\infty)}-\beta I}
   \bigl(2\opnorm{H^{(\infty)}}+\zeta_k\bigr)
   \zeta_k\norm{\ket{x}}_2^2.
\end{aligned}
\end{equation}
The last step uses the triangle inequality $\opnorm{H^{(k)}}\le\opnorm{H^{(\infty)}}+\zeta_k$. Since $\zeta_k\to0$, the coefficient of $\norm{\ket{x}}_2^2$ is at most $\beta\eta^2$ for all $k\ge j$, with $j$ independent of $\ket{x}$. Combining this estimate with Eq.~\eqref{eq:form-expansion} yields $\psi(H^{(k)}\ket{x})\ge\beta^2\psi(\ket{x})$. Each such step therefore preserves positivity of $\psi$ and multiplies its positive value by at least $\beta^2$.

Let $\ket{\mu}$ be a unit eigenvector of $H^{(\infty)}$ for $\lambda$, so that $\psi(\ket{\mu})=\lambda-\beta>0$, and set $\ket{y}\coloneqq\bigl(Q^{(j)}\bigr)^{-1}\ket{\mu}$. Then $\psi(Q^{(k)}\ket{y})\ge\beta^{2(k-j)}(\lambda-\beta)$ grows exponentially for $k\ge j$, whereas $\psi(Q^{(k)}\ket{y})\le\opnorm{H^{(\infty)}-\beta I}\opnorm{Q^{(k)}}^2\norm{\ket{y}}_2^2$ is bounded. This contradiction proves $H^{(\infty)}\preceq I$.
\end{proof}

\begin{theorem}[Global convergence]
\label{thm:global}
Let $C=C^\dagger$, let $s\ge0$ be an admissible shift (Eq.~\eqref{eq:admissible-shift}), and let $X^{(0)},S^{(0)}\succeq0$ satisfy Eq.~\eqref{eq:support-condition}. Then the generalized RW iteration in Eq.~\eqref{eq:rw} is well defined for all $k\ge0$, with $Y^{(k)}\succ0$ and the normalization identity in Eq.~\eqref{eq:pair-normalization} holding at every step. The sequences $X^{(k)}$, $S^{(k)}$, and $Y^{(k)}$ converge to limits $X^{(\infty)},S^{(\infty)}\succeq0$ and $Y^{(\infty)}\succ0$. The matrix $\Zopt\coloneqq Y^{(\infty)}-sI_B$ is the unique optimizer of the normalized dual problem in Eq.~\eqref{eq:dual}, and
\begin{subequations}\label{eq:limit-certificate}
\begin{align}
 \tr_A X^{(\infty)}+S^{(\infty)}&=I_B,\label{eq:limit-normalized}\\
 I_A\otimes\Zopt&\succeq C,\quad\Zopt\succeq0,\label{eq:limit-dual-feasibility}\\
 \tr(CX^{(\infty)})&=\fopt,\label{eq:limit-primal-value}\\
 \tr\Zopt&=\fopt.\label{eq:limit-dual-value}
\end{align}
\end{subequations}
The limiting pair satisfies complementary slackness,
\begin{equation}\label{eq:limit-complementarity}
\begin{split}
 (I_A\otimes\Zopt-C)X^{(\infty)}&=0,\\
 \Zopt S^{(\infty)}&=0.
\end{split}
\end{equation}
\end{theorem}
\begin{proof}
The induction preceding Eq.~\eqref{eq:augmented-update} gives positive definite normalizers and pair normalization at every step. Lemmas~\ref{lem:finite-length} and~\ref{lem:nonsingular} prove that $\Gamma^{(k)}$ and $Y^{(k)}$ converge, with $Y^{(\infty)}\succ0$. Rewriting Eq.~\eqref{eq:factor-recursion} using Eq.~\eqref{eq:compressed-w} gives
\begin{equation}\label{eq:factor-reconstruction}
 G^{(k+1)}=\bigl(\Gamma^{(k)}\bigr)^\dagger C_+^{1/2}
 (I_{A_+}\otimes\bigl(Y^{(k)}\bigr)^{-1}).
\end{equation}
Because $\Gamma^{(k)}$ converges and $Y^{(\infty)}\succ0$ ensures convergence of $(Y^{(k)})^{-1}$, the right-hand side of Eq.~\eqref{eq:factor-reconstruction} has a limit. It follows that $G^{(k)}$, $V^{(k)}=\hat\Theta(G^{(k)})$, and $X_+^{(k)}=(G^{(k)})^\dagger G^{(k)}$ converge. In particular, the diagonal blocks $X^{(k)}$ and $S^{(k)}$ have positive semidefinite limits. Continuity gives Eq.~\eqref{eq:limit-normalized} and $\tr(CX^{(\infty)})=f^{(\infty)}$. This reconstruction does not invert $C_+$ and is therefore valid even when it is singular.

By Eqs.~\eqref{eq:h-convergence} and~\eqref{eq:q-bounded} of Lemma~\ref{lem:nonsingular}, the operators $H^{(k)}$ and the invertible auxiliary matrices $Q^{(k)}$ satisfy the hypotheses of Lemma~\ref{lem:cone} on $\KCal$. Thus $H^{(\infty)}\preceq I_{\KCal}$. Since $C_+^{1/2}=\Pi_{\KCal}C_+^{1/2}\Pi_{\KCal}$, extension by zero on $\KCal^\perp$ gives
\begin{equation}\label{eq:limit-contraction}
 C_+^{1/2}(I_{A_+}\otimes\bigl(Y^{(\infty)}\bigr)^{-1})C_+^{1/2}
 \preceq I_{A_+}\otimes I_B.
\end{equation}
For $L\coloneqq(I_{A_+}\otimes\bigl(Y^{(\infty)}\bigr)^{-1/2})C_+^{1/2}$, this is $L^\dagger L\preceq I_{A_+}\otimes I_B$. The equivalent inequality $LL^\dagger\preceq I_{A_+}\otimes I_B$, after congruence with $I_{A_+}\otimes\bigl(Y^{(\infty)}\bigr)^{1/2}$, yields
\begin{equation}\label{eq:augmented-dual-feasible}
 I_{A_+}\otimes Y^{(\infty)}\succeq C_+.
\end{equation}
With $\Zopt\coloneqq Y^{(\infty)}-sI_B$, the block identity
\begin{equation}\label{eq:dual-block-identity}
\begin{split}
 I_{A_+}\otimes Y^{(\infty)}-C_+
 &=(I_A\otimes\Zopt-C)\oplus\Zopt\succeq0
\end{split}
\end{equation}
proves both constraints of the normalized dual problem in Eq.~\eqref{eq:dual}. Weak duality for the normalized problems, the limit of Eq.~\eqref{eq:original-ascent}, and primal feasibility give
\begin{equation}\label{eq:sandwich}
 \fopt\le\tr\Zopt=\tr Y^{(\infty)}-s\,d_B\le f^{(\infty)}\le\fopt.
\end{equation}
The chain in Eq.~\eqref{eq:sandwich} begins and ends at $\fopt$, so every inequality is an equality. This establishes Eqs.~\eqref{eq:limit-primal-value} and~\eqref{eq:limit-dual-value} and identifies the residual in Eq.~\eqref{eq:gap} as
\begin{equation}\label{eq:original-objective-gap}
 r^{(k)}=f_+^{(\infty)}-f_+^{(k)}=\fopt-f^{(k)}
 \qquad(k\ge1).
\end{equation}
The same identity holds at $k=0$ for a normalized initial pair.

Equality in Eq.~\eqref{eq:weak-duality} with $D=I_B$ makes both of its nonnegative terms vanish, so that
\begin{subequations}\label{eq:limit-trace-complementarity}
\begin{align}
 \tr\bigl[(I_A\otimes\Zopt-C)X^{(\infty)}\bigr]&=0,\label{eq:limit-trace-complementarity-primal}\\
 \tr\bigl[\Zopt(I_B-\tr_A X^{(\infty)})\bigr]=\tr(\Zopt S^{(\infty)})&=0.\label{eq:limit-trace-complementarity-slack}
\end{align}
\end{subequations}
Since the factors in these zero-trace products are positive semidefinite, Eq.~\eqref{eq:limit-trace-complementarity} implies Eq.~\eqref{eq:limit-complementarity}.

For uniqueness, let $Z$ be any other dual optimizer. The same equality in weak duality gives $(I_A\otimes Z-C)X^{(\infty)}=0$ and $ZS^{(\infty)}=0$. Subtracting Eq.~\eqref{eq:limit-complementarity} and taking the partial trace in the first relation yields
\begin{equation}\label{eq:dual-uniqueness-kernels}
\begin{split}
 (Z-\Zopt)\tr_AX^{(\infty)}&=0,\\
 (Z-\Zopt)S^{(\infty)}&=0.
\end{split}
\end{equation}
Adding these relations and using Eq.~\eqref{eq:limit-normalized} gives $Z=\Zopt$.
\end{proof}

The primal optimizer need not be unique, and the original dual optimizer $\Zopt$ can be singular even though $Y^{(\infty)}\succ0$. For $C\succeq0$ with $\tr_AC\succ0$, take $s=0$ and $S^{(0)}=0$. The support condition then reduces to $\ker X^{(0)}\cap\supp(C)=\{0\}$, and the theorem gives the original RW convergence result in Eq.~\eqref{eq:rw-psd}, including singular cost matrices and admissible singular seeds. No slack matrix needs to be stored in this specialization.

The channel and measurement applications use this zero-shift specialization. Their original normalization is recovered because $S^{(k)}=0$.

\begin{corollary}[RW iteration for entanglement fidelity]
\label{cor:rw}
Let $\rho$ be an input state and $\ECal$ a noise channel with $\ECal(\rho)\succ0$, and let $X^{(0)}\coloneqq J((\RCal^{(0)})^\dagger)\succeq0$ satisfy $\ker X^{(0)}\cap\supp(C_{\rho,\ECal})=\{0\}$.
Then the RW iteration of Eq.~\eqref{eq:rw-psd} converges to a channel maximizing $\Fe(\rho,\RCal\circ\ECal)$.
\end{corollary}
\begin{proof}
The fidelity is the linear objective in Eq.~\eqref{eq:primal} with the cost matrix $C_{\rho,\ECal}$ of Eq.~\eqref{eq:fidelity-cost}. Its partial trace $\ECal(\rho^2)$ is positive definite because $\ker\ECal(\rho^2)=\ker\ECal(\rho)$, as established after Eq.~\eqref{eq:fidelity-marginal}. Taking $s=0$ and $S^{(0)}=0$, the hypothesis on $X^{(0)}$ gives Eq.~\eqref{eq:support-condition} for this cost matrix.
Theorem~\ref{thm:global} therefore applies and gives a normalized globally optimal limit.
\end{proof}

\begin{corollary}[JRF iteration for mixed state ensembles]
\label{cor:jrf}
Let $\{(p_i,\rho_i)\}_{i=1}^m$ be an ensemble with $\sum_i p_i\rho_i\succ0$, and let the initial operators $\{M_i^{(0)}\}_{i=1}^m$ be positive semidefinite with $\ker M_i^{(0)}\cap\supp(\rho_i)=\{0\}$ for every $i$.
Then the JRF iteration of Eq.~\eqref{eq:jrf} converges to a POVM maximizing $\Psucc(\{(p_i,\rho_i)\}_i)$.
\end{corollary}
\begin{proof}
Use the cost matrix in Eq.~\eqref{eq:cq-cost}, whose partial trace is positive definite by Eq.~\eqref{eq:ensemble-marginal}, and the block diagonal seed
$X^{(0)}\coloneqq \sum_i|i\rangle\langle i|\otimes M_i^{(0)}$.
Since $p_i>0$ gives $\supp(\sigma_i)=\supp(\rho_i)$, we have $\ker X^{(0)}=\bigoplus_i\ker M_i^{(0)}$ and $\supp(C_{\mathrm{cq}})=\bigoplus_i\supp(\rho_i)$, so that the hypothesis on the family gives Eq.~\eqref{eq:support-condition} at $s=0$ with $S^{(0)}=0$.
The RW iteration preserves these blocks and acts on them as
Eq.~\eqref{eq:jrf}, and Theorem~\ref{thm:global} applies.
\end{proof}

The uniform identity POVM satisfies the initialization hypothesis, and so do the singular seeds $M_i^{(0)}=\sigma_i^{\alpha-2}$ of Section~\ref{sec:iterations}.
The JRF iteration started at the pretty good measurement, or at any power PGM of Eq.~\eqref{eq:pgm}, therefore converges as well,
since its iterates form the tail of the sequence generated from such a seed.
If the ensemble average is singular, the corollary applies on its
support. Any normalized extension of the limiting measurement to the
orthogonal complement preserves its success probability.

\section{Iteration Bounds}
\label{sec:complexity}

This section bounds the number of iterations as $\varepsilon\to0$, with the dimensions $d_A,d_B$, the cost matrix $C$, the admissible shift $s$, and the seed pair $X^{(0)},S^{(0)}$ held fixed. We assume the hypotheses of Theorem~\ref{thm:global} and use the notation of Section~\ref{sec:convergence}. As shown in the proof of Lemma~\ref{lem:finite-length}, there exist constants $\xi>0$ and $\theta\in[1/2,1)$ such that the \L{}ojasiewicz gradient inequality for $f_+$ holds along all sufficiently late iterates. We fix one such pair, whose values are not determined explicitly by the general argument. By global optimality and the objective offset, the residual in Eq.~\eqref{eq:gap} satisfies
\begin{equation}\label{eq:objective-gap-optimality}
 r^{(k)}=f_+^{(\infty)}-f_+^{(k)}=\fopt-f^{(k)}\qquad(k\ge1).
\end{equation}
The constants and the accuracy range of the following bounds may depend on the dimensions, cost matrix, shift, and initialization. For a fixed general $D\ne0$, the inverse substitutions in Eq.~\eqref{eq:constraint-transformation} preserve objective and reduced dual certificate gaps, while bounding the original Frobenius matrix error by $\opnorm{D}$ times its normalized counterpart. The same accuracy exponents therefore apply, with constants that may also depend on $D$.

\begin{table}[t]
 \centering
 \caption{Iteration numbers sufficient for an absolute accuracy $\varepsilon$ in the generalized RW iteration, as $\varepsilon\to0$ with the dimensions, cost matrix, shift, and initialization fixed. Lemma~\ref{lem:iteration-number} proves the three rows using Eq.~\eqref{eq:rate-recurrence}, Lemma~\ref{lem:certificate}, and Eq.~\eqref{eq:matrix-error}, respectively.}
 \label{tab:complexity}
 \small
 \begin{tabular}{@{}lcc@{}}
  \toprule
  Target & $\theta=\frac{1}{2}$ & $\frac{1}{2}<\theta<1$\\
  \midrule
  $r^{(k)}\le\varepsilon$ & $\OCal\bigl(\log\frac{1}{\varepsilon}\bigr)$ & $\OCal\bigl(\varepsilon^{-(2\theta-1)}\bigr)$\\[2pt]
  $r_{\mathrm{cert}}^{(k)}\le\varepsilon$ & $\OCal\bigl(\log\frac{1}{\varepsilon}\bigr)$ & $\OCal\bigl(\varepsilon^{-2(2\theta-1)}\bigr)$\\[2pt]
  $\fnorm{X^{(k)}-X^{(\infty)}}\le\varepsilon$ & $\OCal\bigl(\log\frac{1}{\varepsilon}\bigr)$ & $\OCal\bigl(\varepsilon^{-\frac{2\theta-1}{1-\theta}}\bigr)$\\
  \bottomrule
 \end{tabular}
\end{table}

We consider the objective gap $r^{(k)}$, a computable certificate gap, and the matrix error $\fnorm{X^{(k)}-X^{(\infty)}}$. The matrix $Y^{(k)}-sI_B$ need not satisfy either constraint of the normalized dual problem in Eq.~\eqref{eq:dual}. We therefore define
\begin{subequations}\label{eq:certificate}
\begin{align}
 \tau^{(k)}&\coloneqq\max\left\{\begin{aligned}
 &0,\ \lambda_{\max}(\widetilde C-I_A\otimes Y^{(k)}),\\[-2pt]
 &s-\lambda_{\min}(Y^{(k)})
 \end{aligned}\right\},\label{eq:certificate-shift}\\
 \widetilde{Y}^{(k)}&\coloneqq Y^{(k)}+\tau^{(k)}I_B,\label{eq:certificate-dual}\\
 Z^{(k)}&\coloneqq\widetilde{Y}^{(k)}-sI_B,\label{eq:certificate-original-dual}\\
 r_{\mathrm{cert}}^{(k)}&\coloneqq\tr Z^{(k)}-f^{(k)}.\label{eq:certificate-gap}
\end{align}
\end{subequations}
The two spectral terms enforce $I_A\otimes Z^{(k)}\succeq C$ and $Z^{(k)}\succeq0$, respectively, so $Z^{(k)}$ is feasible for the normalized dual problem in Eq.~\eqref{eq:dual} for every $k\ge1$. Weak duality therefore ensures that $r_{\mathrm{cert}}^{(k)}\le\varepsilon$ certifies an objective error of at most $\varepsilon$ in exact arithmetic.

\begin{lemma}[Certificate]
\label{lem:certificate}
For every $k\ge1$,
\begin{equation}\label{eq:certificate-bounds}
 r^{(k)}\le r_{\mathrm{cert}}^{(k)}\le r^{(k)}+c_{\mathrm{cert}}\sqrt{r^{(k)}},
\end{equation}
where $c_{\mathrm{cert}}$ depends only on the dimensions, $C$, and $s$.
\end{lemma}
\begin{proof}
Weak duality gives the first bound. To prove the second, set
\begin{equation}\label{eq:augmented-dual-slack}
 \Delta_+\coloneqq I_{A_+}\otimes Y^{(\infty)}-C_+\succeq0
\end{equation}
and abbreviate $X_+\coloneqq X_+^{(k)}$ and $Y\coloneqq Y^{(k)}$. Pair normalization gives $\tr_{A_+}X_+=I_B$ and $\tr X_+=d_B$. Since $\tr Y^{(\infty)}=\fopt+s\,d_B$, the objective offset yields $\tr(\Delta_+X_+)=r^{(k)}$. The bounds $X_+^2\preceq d_BX_+$ and $\Delta_+^2\preceq\opnorm{\Delta_+}\Delta_+$ therefore imply
\begin{equation}\label{eq:slack-bound}
 \fnorm{X_+\Delta_+}=\sqrt{\tr(\Delta_+X_+^2\Delta_+)}
 \le\sqrt{d_B\opnorm{\Delta_+}\,r^{(k)}}.
\end{equation}
Inserting $C_+=I_{A_+}\otimes Y^{(\infty)}-\Delta_+$ into $Y^2=\tr_{A_+}(C_+X_+C_+)$ gives
\begin{equation}\label{eq:normalizer-squared-difference}
\begin{aligned}
 Y^2-\bigl(Y^{(\infty)}\bigr)^2
 &=-\tr_{A_+}(\Delta_+X_+C_+)\\
 &\quad-\tr_{A_+}\bigl[(I_{A_+}\otimes Y^{(\infty)})X_+\Delta_+\bigr].
\end{aligned}
\end{equation}
A partial trace over $A_+$ increases the Frobenius norm by at most $\sqrt{d_{A_+}}$. Consequently,
\begin{equation}\label{eq:normalizer-squared-bound}
\begin{aligned}
 &\fnorm{Y^2-\bigl(Y^{(\infty)}\bigr)^2}\\
 &\quad\le\sqrt{d_{A_+}}\bigl(\opnorm{C_+}+\opnorm{Y^{(\infty)}}\bigr)
       \fnorm{X_+\Delta_+}.
\end{aligned}
\end{equation}
Moreover,
\begin{equation}\label{eq:normalizer-sylvester-identity}
 Y^2-\bigl(Y^{(\infty)}\bigr)^2=Y(Y-Y^{(\infty)})+(Y-Y^{(\infty)})Y^{(\infty)}.
\end{equation}
Taking its real Frobenius inner product with $Y-Y^{(\infty)}$, using $Y\succeq0$ and $Y^{(\infty)}\succ0$, yields
\begin{equation}\label{eq:normalizer-root-bound}
 \lambda_{\min}(Y^{(\infty)})\fnorm{Y-Y^{(\infty)}}
 \le\fnorm{Y^2-\bigl(Y^{(\infty)}\bigr)^2}.
\end{equation}
The correction in Eq.~\eqref{eq:certificate-shift} is precisely
\begin{equation}\label{eq:certificate-shift-augmented}
 \tau^{(k)}=\max\{0,\lambda_{\max}(C_+-I_{A_+}\otimes Y)\}.
\end{equation}
Since $C_+-I_{A_+}\otimes Y\preceq I_{A_+}\otimes(Y^{(\infty)}-Y)$, it follows that $\tau^{(k)}\le\opnorm{Y-Y^{(\infty)}}$. Lemma~\ref{lem:ascent} gives $\tr Y\le f_+^{(k+1)}\le\fopt+s\,d_B$, so $r_{\mathrm{cert}}^{(k)}\le r^{(k)}+d_B\tau^{(k)}$. Combining these estimates with Eqs.~\eqref{eq:normalizer-squared-bound} and~\eqref{eq:normalizer-root-bound} and the explicit $\sqrt{r^{(k)}}$ bound in Eq.~\eqref{eq:slack-bound}, proves Eq.~\eqref{eq:certificate-bounds} with
\begin{equation}\label{eq:certificate-constant}
 c_{\mathrm{cert}}\coloneqq\sqrt{d_{A_+}d_B^3\opnorm{\Delta_+}}\frac{\opnorm{C_+}+\opnorm{Y^{(\infty)}}}{\lambda_{\min}(Y^{(\infty)})}.
\end{equation}
The denominator concerns the positive definite limiting normalizer, not the possibly singular original dual optimizer $\Zopt=Y^{(\infty)}-sI_B$. Uniqueness of the dual optimizer makes the constant independent of the initialization.
\end{proof}

\begin{lemma}[Iteration number]
\label{lem:iteration-number}
With the dimensions, cost matrix, shift, and seed pair held fixed, the iteration numbers in Table~\ref{tab:complexity} suffice for the respective targets as $\varepsilon\to0$.
\end{lemma}
\begin{proof}
A sequence that becomes stationary reaches all three targets in finitely many steps. Otherwise, for every sufficiently large $k$, we have $r^{(k)}\le1$. Lemmas~\ref{lem:ascent} and~\ref{lem:gradient-bound}, together with Eq.~\eqref{eq:lojasiewicz} for $f_+$ at $V=V^{(k+1)}$, give
\begin{equation}\label{eq:rate-recurrence}
 r^{(k)}-r^{(k+1)}\ge\fnorm{\delta\Gamma^{(k)}}^2
 \ge\frac{\bigl(r^{(k+1)}\bigr)^{2\theta}}{\xi^2\kappa^2}.
\end{equation}
If $\theta=1/2$, this reads $r^{(k+1)}\le r^{(k)}/(1+\xi^{-2}\kappa^{-2})$, so the gap contracts exponentially. If $\theta>1/2$, fix a level $\omega\in(0,1]$ and start with a gap of at most $2\omega$. By Eq.~\eqref{eq:rate-recurrence}, each step ending above $\omega$ decreases the gap by at least $\omega^{2\theta}/(\xi^2\kappa^2)$. The sum of these decreases cannot exceed $2\omega$, so there are at most $2\xi^2\kappa^2\omega^{1-2\theta}$ such steps. Including the final step that reaches or crosses $\omega$ gives the bound $1+2\xi^2\kappa^2\omega^{1-2\theta}$.

Now let $j_\varepsilon\coloneqq\lceil\log_2(1/\varepsilon)\rceil$ and successively apply this bound at levels $\omega=2^{-j}$, for $j=1,\ldots,j_\varepsilon$. Beyond a fixed initial segment where the recurrence need not hold, the number of steps needed to reach $r^{(k)}\le\varepsilon$ is at most
\begin{equation}\label{eq:dyadic-iteration-count}
\begin{aligned}
 &\sum_{j=1}^{j_\varepsilon}\bigl(1+2\xi^2\kappa^2 2^{j(2\theta-1)}\bigr)\\
 &\qquad=\OCal\bigl(2^{j_\varepsilon(2\theta-1)}\bigr)
 =\OCal\bigl(\varepsilon^{-(2\theta-1)}\bigr).
\end{aligned}
\end{equation}
Here $2\theta-1>0$ makes the geometric sum comparable to its last term, which also dominates $j_\varepsilon$. The last equality follows from $\varepsilon/2<2^{-j_\varepsilon}\le\varepsilon$.

For the matrix error, Eq.~\eqref{eq:finite-length} and $\fnorm{\delta\Gamma^{(k)}}^2\le r^{(k)}$ give
\begin{equation}\label{eq:rate-tail}
 \fnorm{\Gamma^{(\infty)}-\Gamma^{(k+1)}}\le\sqrt{r^{(k)}}+\frac{\xi\kappa}{1-\theta}\bigl(r^{(k+1)}\bigr)^{1-\theta}.
\end{equation}
This is $\OCal((r^{(k)})^{1-\theta})$, because the gap is nonincreasing and $\sqrt r\le r^{1-\theta}$ for $0\le r\le1$. By Eq.~\eqref{eq:z-positive}, the limiting normalizer is positive definite. Eqs.~\eqref{eq:factor-update}, \eqref{eq:compressed-w}, and~\eqref{eq:normalizer-finite} therefore express $X_+^{(k+2)}$ as a smooth function of $\Gamma^{(k+1)}$ near its limit. This function maps $\Gamma^{(\infty)}$ to $X_+^{(\infty)}$ and is locally Lipschitz continuous. Restricting to the first diagonal block gives the following bound for some constant $c_{\mathrm{Lips}}>0$ and all sufficiently large $k$.
\begin{equation}\label{eq:matrix-error}
\begin{aligned}
 \fnorm{X^{(k+2)}-X^{(\infty)}}
 &\le\fnorm{X_+^{(k+2)}-X_+^{(\infty)}}\\
 &\le c_{\mathrm{Lips}}\fnorm{\Gamma^{(\infty)}-\Gamma^{(k+1)}}\\
 &=\OCal\bigl((r^{(k)})^{1-\theta}\bigr).
\end{aligned}
\end{equation}
Lemma~\ref{lem:certificate} makes the certificate gap at most $\varepsilon$ once the objective gap reaches a sufficiently small constant multiple of $\varepsilon^2$. Likewise, Eq.~\eqref{eq:matrix-error} makes the matrix error at most $\varepsilon$ two steps after the gap reaches a sufficiently small constant multiple of $\varepsilon^{1/(1-\theta)}$. Substituting these tolerances into the first row multiplies the exponent $2\theta-1$ by $2$ and $1/(1-\theta)$, respectively, while preserving a logarithmic bound.
\end{proof}

Since $\theta<1$, Lemma~\ref{lem:iteration-number} yields the coarser bounds of $\OCal(\varepsilon^{-1})$ iterations for reaching $r^{(k)}\le\varepsilon$ and $\OCal(\varepsilon^{-2})$ for reaching $r_{\mathrm{cert}}^{(k)}\le\varepsilon$. These estimates apply as $\varepsilon\to0$ with the instance, shift, and seed pair fixed.

The dense arithmetic costs per iteration, including certificate evaluation, are $\OCal(d_A^3d_B^3)$ for the generalized RW iteration and $\OCal(d_Ad_B^3)$ for the JRF iteration. The additional slack block requires only $\OCal(d_B^3)$ operations. The constants and onset of the iteration bounds may depend on the dimensions and shift, so Table~\ref{tab:complexity} does not establish total complexity as the dimensions vary.

Diagonal cost matrices provide a representative case of exponential convergence. For a seed pair diagonal in the corresponding product and slack bases and satisfying Eq.~\eqref{eq:support-condition}, Appendix~\ref{sec:diagonal-rates} proves that the local gradient inequality holds with $\theta=1/2$ on the full augmented manifold. The corresponding column of Table~\ref{tab:complexity} therefore gives $\OCal(\log(1/\varepsilon))$ iteration bounds for the objective gap, certificate gap, and matrix error. These errors admit geometric decay bounds, including the possibility of convergence in finitely many steps.

The density peaks in Figure~\ref{fig:gue} are also consistent with $\OCal(\log(1/\varepsilon))$ scaling over the displayed tolerance range. These finite samples do not determine the gradient exponent of an individual instance. Objective accuracy is generally reached sooner than certificate or matrix accuracy, in accordance with Lemma~\ref{lem:certificate} and Eq.~\eqref{eq:matrix-error}.

Polynomial convergence can nevertheless occur even when the seed satisfies the support condition. Appendix~\ref{sec:slow-trajectories} constructs a family indexed by $n\ge2$ for which the objective and certificate iteration counts are $\Theta(\varepsilon^{-(1-2^{1-n})})$, while the matrix iteration count is $\Theta(\varepsilon^{-(2^n-2)})$. As $n$ increases, the objective power approaches one and the exponent governing the matrix iteration count grows without bound. Thus no smaller fixed power can replace the coarse $\OCal(\varepsilon^{-1})$ objective bound across all members of this family. This slow convergence arises because the objective becomes extremely flat along these iterations near the optimum, leading to very small updates.

\section{Discussion}
\label{sec:discussion}
We have developed a generalized RW iteration for maximizing $\tr(CX)$ over positive semidefinite matrices subject to $\tr_A X\preceq D$, with arbitrary Hermitian $C$ and positive semidefinite $D$. Under the stated shift and support conditions, the iterates converge to a global optimum, and our analysis quantifies convergence in the objective, certificate, and matrix errors. Figure~\ref{fig:timing} compares the time required to reach a common objective accuracy assessed against a precomputed dual reference. Under the timing conventions of Appendix~\ref{sec:runtime-details}, the original RW and JRF implementations achieve much lower median runtimes than the other representative approaches tested.

Extending the construction to general linear constraints is a natural direction, but the tensor structure of the partial trace is essential to our normalization. A general linear map need not admit an analogous congruence that restores feasibility while preserving ascent. Such an extension would require new ingredients in both the update and its convergence proof, and lies beyond the scope of this work.

Within the present constraint class, understanding how the cost matrix, spectral shift $s$, and initialization affect convergence could guide practical algorithm design. Admissible shifts preserve the optimization problem but can change convergence speed. A suitable feasible seed pair satisfying the support condition may likewise help the iteration approach the optimum more rapidly. Studying these choices together is a natural direction, while our analysis treats both the shift and initialization as fixed.

\section*{Note added}
We acknowledge an independent proof by Yuxuan Zhang~\cite{Zhang2026AgenticProofs} that the JRF iteration initialized by the uniform POVM converges to an optimal measurement. This result is also a special case of Corollary~\ref{cor:jrf}.

\begin{acknowledgments}
B.L. thanks Yat Wong and Guo Zheng for helpful discussions and comments. B.L. and L.J. acknowledge support from the ARO (W911NF-23-1-0077), ARO MURI (W911NF-21-1-0325), AFOSR MURI (FA9550-21-1-0209, FA9550-23-1-0338), ONR MURI (N000142612102), DARPA (HR0011-24-9-0361), NSF (OSI-2326767, OSI-2426975, EEC-2550064, OSI-2553578, OSI-2553619). This material is based upon work supported by the U.S. Department of Energy, Office of Science, National Quantum Information Science Research Centers and Advanced Scientific Computing Research (ASCR) program under contract number DE-AC02-06CH11357 as part of the InterQnet quantum networking project.
S.R. and W.G. acknowledge support by ASTAR (M24M8b0004), Singapore National Research Foundation (NRF-CRP30-2023-0003, NRF-CRP31-0001, NRF2023-ITC004-001 and NRF-MSG-2023-0002) and Singapore Ministry of Education Tier 2 Grant (MOE-T2EP50222-0018).
We acknowledge the use of generative AI tools to support exploration of ideas and the typesetting of this work.
The numerical computational work was completed with the resources provided by the University of Chicago's Research Computing Center.
\end{acknowledgments}

\appendix
\section{The half gradient exponent for diagonal cost matrices}
\label{sec:diagonal-rates}
Fix an orthonormal product basis $\ket{ab}\coloneqq\ket a_A\otimes\ket b_B$. Suppose that the Hermitian cost $C$ and the seed $X^{(0)}\succeq0$ are diagonal in this basis, while $S^{(0)}\succeq0$ is diagonal in $\{\ket b_B\}$. Choose a fixed admissible shift satisfying Eq.~\eqref{eq:admissible-shift} and a seed pair satisfying Eq.~\eqref{eq:support-condition}. We prove that the local gradient inequality at the limiting factor holds with
\begin{equation}\label{eq:diagonal-half-exponent}
 \theta=\frac12.
\end{equation}
The inequality holds on the full augmented Stiefel manifold $\MCal$, including nondiagonal perturbations. The corresponding column of Table~\ref{tab:complexity} then gives logarithmic iteration bounds for the objective gap, certificate gap, and matrix error. The argument allows singular or indefinite costs and seeds that are not initially normalized.

\subsection{The limiting diagonal factor}
Write $C=\sum_{a,b}c_{ab}\ket{ab}\bra{ab}$, where $1\le a\le d_A$ and $1\le b\le d_B$. In the augmentation of Eq.~\eqref{eq:augmented-data}, label the slack direction by $0$ and let $i,j\in\{0,\ldots,d_A\}$, retaining the slack as the last direct summand. Denote the diagonal entries of $C_+$ by the cost weights $w_{ib}$ and those of $X_+^{(k)}$ by $p_{ib}^{(k)}$. Thus
\begin{equation}\label{eq:diagonal-entries}
\begin{aligned}
 w_{0b}&=s,&p_{0b}^{(k)}&=S_{bb}^{(k)},\\
 w_{ab}&=c_{ab}+s,&p_{ab}^{(k)}&=\bra{ab}X^{(k)}\ket{ab}.
\end{aligned}
\end{equation}
Admissibility makes every weight nonnegative and supplies at least one positive weight for each $b$. The support condition requires $p_{ib}^{(0)}>0$ whenever $w_{ib}>0$.

Equation~\eqref{eq:rw} preserves diagonality. Its scalar recursion and explicit solution are
\begin{equation}\label{eq:diagonal-iterates}
\begin{aligned}
 p_{ib}^{(k+1)}&=\frac{w_{ib}^2p_{ib}^{(k)}}{\sum_jw_{jb}^2p_{jb}^{(k)}}\qquad(k\ge0),\\
 p_{ib}^{(k)}&=\frac{w_{ib}^{2k}p_{ib}^{(0)}}{\sum_jw_{jb}^{2k}p_{jb}^{(0)}}\qquad(k\ge1).
\end{aligned}
\end{equation}
The denominators remain positive because a positive weight with a positive seed entry retains a positive iterate entry. The first update gives $\sum_i p_{ib}^{(k)}=1$ for $k\ge1$, and the explicit formula follows by cancellation of successive normalization factors.

For each $b$, let $m_b$ be the largest cost weight, $\mathcal I_b$ the set of indices attaining it, and $p_{\max,b}$ the total initial weight on that set. Explicitly,
\begin{equation}\label{eq:diagonal-maxima}
\begin{aligned}
 m_b&\coloneqq\max_iw_{ib}=s+\max\{0,c_{1b},\ldots,c_{d_A b}\}>0,\\
 \mathcal I_b&\coloneqq\{i:w_{ib}=m_b\},\qquad
 p_{\max,b}\coloneqq\sum_{i\in\mathcal I_b}p_{ib}^{(0)}>0.
\end{aligned}
\end{equation}
Dividing the scalar solution by $m_b^{2k}$ yields
\begin{equation}\label{eq:diagonal-limit}
 p_{ib}^{(\infty)}=
 \begin{cases}
 p_{ib}^{(0)}/p_{\max,b}>0,&i\in\mathcal I_b,\\
 0,&i\notin\mathcal I_b.
 \end{cases}
\end{equation}
In particular, every maximizing index has a positive limiting entry, even when several weights attain $m_b$. This consequence of the support condition will ensure that the optimal factors form a smooth set locally.

The Gram factors of Eq.~\eqref{eq:factor-update} remain diagonal with nonnegative entries, so $G^{(k)}=(X_+^{(k)})^{1/2}$ and
\begin{equation}\label{eq:diagonal-factor-limit}
 G^{(\infty)}=\operatorname{diag}_{(i,b)}\!\left(\sqrt{p_{ib}^{(\infty)}}\right),
 \qquad V^{(\infty)}=\hat\Theta(G^{(\infty)}).
\end{equation}
We call the matrix column $G\ket{ib}$ maximizing if $i\in\mathcal I_b$ and suboptimal otherwise. The maximizing columns of $G^{(\infty)}$ are nonzero and mutually orthogonal, while its suboptimal columns vanish.

\subsection{The local set of optimal factors}
To control all nearby perturbations, consider the normalized factor manifold
\begin{equation}\label{eq:diagonal-factor-manifold}
 \mathcal G\coloneqq
 \bigl\{G\in\Cbb^{(d_{A_+}d_B)\times(d_{A_+}d_B)}:
 \tr_{A_+}(G^\dagger G)=I_B\bigr\}.
\end{equation}
Equations~\eqref{eq:reshaping}--\eqref{eq:stiefel} identify $\mathcal G$ with $\MCal$ through the Frobenius isometry $\hat\Theta$. We equip both real manifolds with the metric induced by $\operatorname{Re}\finner{\cdot}{\cdot}$, so this identification preserves gradient norms.

For $G\in\mathcal G$, normalization gives $\sum_i\norm{G\ket{ib}}_2^2=1$ for each $b$. The augmented objective is therefore at most $\sum_bm_b$, and $G^{(\infty)}$ attains this value. Thus $\fopt=\sum_b(m_b-s)$. Define $\delta_{ib}\coloneqq m_b-w_{ib}\ge0$ and the factor objective gap
\begin{equation}\label{eq:diagonal-factor-gap}
\begin{aligned}
 r_{\mathrm{fac}}(G)
 &\coloneqq\fopt+sd_B-\tr(C_+G^\dagger G)\\
 &=\sum_{b,i}\delta_{ib}\norm{G\ket{ib}}_2^2.
\end{aligned}
\end{equation}
The gap vanishes exactly when every suboptimal column vanishes. Let $\mathcal L_{\mathrm{diag}}$ be the linear space of factors with these columns zero, and let $\mathcal Z_{\mathrm{diag}}$ be the optimal set in $\mathcal G$,
\begin{equation}\label{eq:diagonal-optimal-set}
\begin{aligned}
 \mathcal L_{\mathrm{diag}}
 &\coloneqq\{G:G\ket{ib}=0\text{ whenever }i\notin\mathcal I_b\},\\
 \mathcal Z_{\mathrm{diag}}&\coloneqq\mathcal G\cap\mathcal L_{\mathrm{diag}}.
\end{aligned}
\end{equation}
Here the factors have the size specified in Eq.~\eqref{eq:diagonal-factor-manifold}. Their remaining columns are arbitrary vectors, so neither set is restricted to diagonal factors.

We next show that $\mathcal Z_{\mathrm{diag}}$ is smooth near $G^{(\infty)}$. On $\mathcal L_{\mathrm{diag}}$, the normalization entries are
\begin{equation}\label{eq:diagonal-normalization-entries}
 \bigl[\tr_{A_+}(G^\dagger G)\bigr]_{bb'}
 =\sum_i(G\ket{ib})^\dagger(G\ket{ib'}).
\end{equation}
When $b\ne b'$ and $\mathcal I_b\cap\mathcal I_{b'}$ is empty, this entry vanishes identically. Retain the diagonal equations and the real and imaginary parts of the off-diagonal equations with $b<b'$ and a shared maximizing index. At $G^{(\infty)}$, scaling one maximizing column by a real factor varies its diagonal normalization entry independently. For a retained off-diagonal entry, choose $i\in\mathcal I_b\cap\mathcal I_{b'}$ and perturb $G^{(\infty)}\ket{ib'}$ by a complex multiple of $G^{(\infty)}\ket{ib}$. Orthogonality makes this vary only the $(b,b')$ entry and its conjugate, with arbitrary first derivatives. These variations stay in $\mathcal L_{\mathrm{diag}}$, and their coefficients are nonzero by Eq.~\eqref{eq:diagonal-limit}.

The real analytic implicit function theorem therefore applies to the remaining normalization equations. It gives a smooth optimal set with tangent space
\begin{equation}\label{eq:diagonal-optimal-tangent}
 T_{G^{(\infty)}}\mathcal Z_{\mathrm{diag}}
 =T_{G^{(\infty)}}\mathcal G\cap\mathcal L_{\mathrm{diag}}.
\end{equation}
This establishes the local geometry even when the maximizing weights are not unique.

\subsection{The gradient inequality and its rate consequences}
Suppose first that $C\ne0$. At least one $\delta_{ib}$ is then positive, since otherwise comparison with the slack weight $w_{0b}=s$ would force every $c_{ab}$ to vanish. For a smooth curve $G(t)\in\mathcal G$ with $G(0)=G^{(\infty)}$ and velocity $U\coloneqq G'(0)$, Eq.~\eqref{eq:diagonal-factor-gap} gives
\begin{equation}\label{eq:diagonal-second-variation}
 \left.\frac{\dd^2}{\dd t^2}r_{\mathrm{fac}}(G(t))\right|_{t=0}
 =2\sum_{b,i}\delta_{ib}\norm{U\ket{ib}}_2^2.
\end{equation}
Terms containing the acceleration vanish because every column with a positive coefficient is zero at the limit. The quadratic form vanishes precisely on the tangent space in Eq.~\eqref{eq:diagonal-optimal-tangent}. It is consequently positive in every nonzero direction transverse to the optimal set.

Choose real local coordinates $(u,v)$ centered at $G^{(\infty)}$, with $u$ describing motion along the optimal set and $v$ describing transverse motion, so that this set is $v=0$. Write the gap in these coordinates as $r_{\mathrm{loc}}(u,v)$. It satisfies $r_{\mathrm{loc}}(u,0)=0$ and $\partial_vr_{\mathrm{loc}}(u,0)=0$, where $\partial_vr_{\mathrm{loc}}$ is the column vector of partial derivatives with respect to the components of $v$. Define the transverse Hessian by
\begin{equation}\label{eq:diagonal-transverse-hessian}
 [H_\perp(u,v)]_{\ell\ell'}
 \coloneqq\frac{\partial^2r_{\mathrm{loc}}}{\partial v_\ell\partial v_{\ell'}}(u,v).
\end{equation}
The second variation makes $H_\perp(0,0)$ positive definite. After shrinking the coordinate neighborhood, continuity gives constants $0<\lambda_-\le\lambda_+$ such that $\lambda_-I_\perp\preceq H_\perp(u,v)\preceq\lambda_+I_\perp$, where $I_\perp$ is the identity on the transverse coordinate space.

Integrating along the coordinate segment from $(u,0)$ to $(u,v)$ yields
\begin{equation}\label{eq:diagonal-coordinate-estimates}
\begin{aligned}
 r_{\mathrm{loc}}(u,v)
 &=\int_0^1(1-t)v^{\mathsf T}H_\perp(u,tv)v\,\dd t\\
 &\le\frac{\lambda_+}{2}\norm{v}_2^2,\\
 v^{\mathsf T}\partial_vr_{\mathrm{loc}}(u,v)
 &=\int_0^1v^{\mathsf T}H_\perp(u,tv)v\,\dd t\\
 &\ge\lambda_-\norm{v}_2^2.
\end{aligned}
\end{equation}
Cauchy--Schwarz gives $\norm{\partial_vr_{\mathrm{loc}}}_2\ge\lambda_-\norm{v}_2$. The coordinate map and its inverse have bounded derivatives locally, so there is $c_{\mathrm{chart}}>0$ such that
\begin{equation}\label{eq:diagonal-gradient-lower-bound}
 \fnorm{\nabla_{\mathcal G}r_{\mathrm{fac}}(G)}
 \ge c_{\mathrm{chart}}\norm{\partial_vr_{\mathrm{loc}}}_2
 \ge c_{\mathrm{chart}}\lambda_-\norm{v}_2.
\end{equation}
Combining this lower bound with the quadratic gap estimate and using the isometry $\hat\Theta$ proves, for every $V\in\MCal$ sufficiently near $V^{(\infty)}$,
\begin{equation}\label{eq:diagonal-gradient-inequality}
\begin{aligned}
 \bigl(\fopt+sd_B-f_+(V)\bigr)^{1/2}
 &\le\xi\fnorm{\nabla_{\MCal}f_+(V)},\\
 \xi&\coloneqq\frac{\sqrt{\lambda_+/2}}{c_{\mathrm{chart}}\lambda_-}.
\end{aligned}
\end{equation}
If $C=0$, admissibility requires $s>0$ and $C_+=s(I_{A_+}\otimes I_B)$. Normalization then makes $f_+$ identically $sd_B$ on $\MCal$, so both sides vanish and $\theta=1/2$ is also permissible. Thus Eq.~\eqref{eq:diagonal-half-exponent} applies in every case under consideration.

Since $V^{(k)}\to V^{(\infty)}$, the half-exponent inequality holds along all sufficiently late iterates. Equation~\eqref{eq:rate-recurrence} consequently gives $r^{(k+1)}\le q_{\mathrm{loc}}r^{(k)}$ for some fixed $0<q_{\mathrm{loc}}<1$. Lemma~\ref{lem:certificate} and Eq.~\eqref{eq:matrix-error} then imply
\begin{equation}\label{eq:diagonal-geometric-bounds}
\begin{aligned}
 r^{(k)}&=\OCal(q_{\mathrm{loc}}^k),\\
 r_{\mathrm{cert}}^{(k)},\ \fnorm{X^{(k)}-X^{(\infty)}}
 &=\OCal(q_{\mathrm{loc}}^{k/2}).
\end{aligned}
\end{equation}
All three errors therefore admit geometric decay bounds. Let $k_{\mathrm{obj}}(\varepsilon)$, $k_{\mathrm{cert}}(\varepsilon)$, and $k_X(\varepsilon)$ denote the first indices $k\ge1$ at which the respective errors are at most $\varepsilon$. The $\theta=1/2$ column of Table~\ref{tab:complexity} gives
\begin{equation}\label{eq:diagonal-hitting-times}
 k_{\mathrm{obj}}(\varepsilon),\ k_{\mathrm{cert}}(\varepsilon),\ k_X(\varepsilon)
 =\OCal\bigl(\log(1/\varepsilon)\bigr).
\end{equation}
These bounds hold as $\varepsilon\to0$ with the dimensions, cost, shift, and seed pair fixed, and also cover iterations that reach their limit in finitely many steps. Diagonality is required in the specified product basis, since an arbitrary unitary change of basis need not preserve the partial trace constraint.

\section{Slow convergence from seeds satisfying the support condition}
\label{sec:slow-trajectories}
We construct a family satisfying Eq.~\eqref{eq:support-condition} whose convergence becomes progressively slower as its dimension increases. The key feature is that the objective becomes extremely flat along the iteration near the optimum, leading to very small updates. For each integer $n\ge2$, define
\begin{equation}\label{eq:slow-family-order}
 m_n\coloneqq2^{n-1}.
\end{equation}
\begin{proposition}\label{prop:slow-family-rates}
Fix $n\ge2$ and use the cost matrix in Eq.~\eqref{eq:slow-family-cost} and the seed in Eq.~\eqref{eq:slow-family-seed}, with $s=0$, $S^{(0)}=0$, and $\chi_{\mathrm{corr}}=1/2$. There exists $t_*(n)>0$ such that every $0<t_0<t_*(n)$ gives a seed satisfying Eq.~\eqref{eq:support-condition}. The resulting RW iteration converges to the matrix in Eq.~\eqref{eq:slow-matrix-limit}, and its first hitting times satisfy Eq.~\eqref{eq:slow-family-iteration-count}. In particular,
\begin{equation}\label{eq:circle-exponent-comparison}
\begin{aligned}
 k_{\mathrm{obj},n}(\varepsilon),\ k_{\mathrm{cert},n}(\varepsilon)
 &=\Theta\!\left(\varepsilon^{-(m_n-1)/m_n}\right),\\
 k_{X,n}(\varepsilon)
 &=\Theta\!\left(\varepsilon^{-(2m_n-2)}\right).
\end{aligned}
\end{equation}
as $\varepsilon\to0^+$ with $n$ and $t_0$ fixed.
\end{proposition}
The three targets are the objective gap, certificate gap, and Frobenius distance to the limiting matrix. The following four subsections prove the proposition by specifying the construction and deriving the scalar dynamics and error estimates along the iteration.

\subsection{Cost matrix and stationary curve}

Take $\HCal_A=\Cbb^{2n+1}$ and $\HCal_B=\Cbb^{n+1}$, with orthonormal bases $\{\ket{a_0}\}\cup\{\ket{a_{j,x}},\ket{a_{j,y}}\}_{j=1}^n$ and $\{\ket{b_j}\}_{j=0}^n$. Define $\ket{p_0}\coloneqq\ket{a_0}\otimes\ket{b_0}$, $\ket{p_j}\coloneqq\ket{a_{j,x}}\otimes\ket{b_j}$, and $\ket{q_j}\coloneqq\ket{a_{j,y}}\otimes\ket{b_j}$. The following projector $P$, vectors $\ket{w_j}$, and coefficient $\gamma_n$ specify the cost matrix $C$.
\begin{equation}\label{eq:slow-family-cost}
\begin{aligned}
 P&\coloneqq\ket{p_0}\bra{p_0}
       +\sum_{j=1}^n(\ket{p_j}\bra{p_j}+\ket{q_j}\bra{q_j}),\\
 \ket{w_0}&\coloneqq\ket{q_1},\qquad
 \gamma_n\coloneqq\frac1{2(3n-2)},\\
 \ket{w_j}&\coloneqq\ket{q_j}+\ket{p_{j+1}}-\ket{p_0}\quad(1\le j<n),\\
 C&\coloneqq P-\gamma_n\sum_{j=0}^{n-1}\ket{w_j}\bra{w_j}.
\end{aligned}
\end{equation}
Since $\sum_j\norm{\ket{w_j}}_2^2=3n-2$, we have $P/2\preceq C\preceq P$. Thus $\supp C=\supp P$, $\rank C=2n+1$, and $\tr_A C\succ0$. We use the admissible shift $s=0$ and slack seed $S^{(0)}=0$, so Eq.~\eqref{eq:rw-psd} applies.

Two invariant subspaces make the construction tractable. Let $\mathcal W$ be spanned by $\ket{p_0}$ and the pairs $(\ket{q_j},\ket{p_{j+1}})$ for $1\le j<n$. Its orthogonal complement within $\supp P$ is $\mathcal U\coloneqq\operatorname{span}\{\ket{p_1},\ket{q_n}\}$, on which $C$ is the identity. Let $\Pi_{\mathcal W}$ denote the orthogonal projector onto $\mathcal W$ and fix $\chi_{\mathrm{corr}}=1/2$. For $\alpha=(\alpha_1,\ldots,\alpha_n)\in\Rbb^n$, introduce a normalized boundary family
\begin{equation}\label{eq:slow-boundary-family}
\begin{aligned}
 \ket{\omega(\alpha)}&\coloneqq\ket{p_0}
 +\sum_{j=1}^{n-1}(\sin\alpha_j\ket{q_j}+\cos\alpha_{j+1}\ket{p_{j+1}}),\\
 B_{\mathrm{iso}}(\alpha)&\coloneqq
 \begin{pmatrix}
 \cos^2\alpha_1&\chi_{\mathrm{corr}}\cos\alpha_1\sin\alpha_n\\
 \chi_{\mathrm{corr}}\cos\alpha_1\sin\alpha_n&\sin^2\alpha_n
 \end{pmatrix},\\
 X_{\mathrm{br}}(\alpha)&\coloneqq
 \ket{\omega(\alpha)}\bra{\omega(\alpha)}\oplus B_{\mathrm{iso}}(\alpha).
\end{aligned}
\end{equation}
The direct sum is taken over $\mathcal W\oplus\mathcal U$, with zero entries outside $\supp P$. Distinct active vectors use distinct $A$ basis labels, so their off-diagonal entries vanish under $\tr_A$. Consequently $\tr_A X_{\mathrm{br}}(\alpha)=I_B$.

Set $\ell_j(\alpha)\coloneqq1-\cos\alpha_{j+1}-\sin\alpha_j$ for $1\le j<n$ and define
\begin{equation}\label{eq:slow-family-gap-function}
 g(\alpha)\coloneqq\sin^2\alpha_1+\sum_{j=1}^{n-1}\ell_j(\alpha)^2.
\end{equation}
The overlaps with $\ket{w_j}$ give $\tr(CX_{\mathrm{br}}(\alpha))=n+1-\gamma_ng(\alpha)$. Since $C\preceq I_{AB}$ bounds every feasible objective by $n+1$, and $X_{\mathrm{br}}(0)$ attains this value, we have $\fopt=n+1$. Thus $\gamma_ng(\alpha)$ is precisely the objective gap on this boundary family.

Write $\alpha=(u,t)$, where $u\coloneqq(\alpha_1,\ldots,\alpha_{n-1})$ and $t\coloneqq\alpha_n$. These are coordinates of the same angle vector, and $\partial_\alpha$ denotes differentiation with respect to all its entries. The Hessian in the first $n-1$ coordinates satisfies
\begin{equation}\label{eq:slow-family-hessian}
 \left[\frac{\partial^2g}{\partial\alpha_i\,\partial\alpha_j}(0)\right]_{i,j=1}^{n-1}
 =\operatorname{diag}(4,2,\ldots,2)\succ0.
\end{equation}
The implicit function theorem therefore defines an analytic curve $\bar\alpha(t)\coloneqq(h_n(t),t)$, with $h_n(0)=h_n'(0)=0$, on which $\partial_{\alpha_j}g=0$ for every $j<n$. These stationary equations read
\begin{equation}\label{eq:slow-family-critical}
 \ell_1=\sin\alpha_1,\qquad
 \ell_j=\ell_{j-1}\tan\alpha_j\quad(2\le j<n).
\end{equation}
Working backwards from $t$, the identities $1-\cos\alpha_{j+1}=\sin\alpha_j+\ell_j$ give $\bar\alpha_j=\Theta(t^{2^{n-j}})$ and $\bar\alpha_1\sim2^{-m_n}t^{m_n}$. Only $\sin^2\bar\alpha_1$ and $\ell_1^2$ contribute at leading order, yielding
\begin{equation}\label{eq:slow-family-stationary-gap}
 \gamma_ng(\bar\alpha(t))\sim c_nt^{2m_n},\qquad
 c_n\coloneqq\gamma_n2^{1-2m_n}>0.
\end{equation}

\subsection{Seed and scalar dynamics}

The boundary family alone does not satisfy the support requirement. We therefore choose a small $t_0>0$, put $\hat{\alpha}\coloneqq\bar\alpha(t_0)$ and $\delta_{\mathrm{seed}}\coloneqq t_0^{4m_n}$, and add a positive covariance on $\mathcal W$ before normalizing. Explicitly, set
\begin{subequations}\label{eq:slow-family-seed}
\begin{align}
 \overline X_{\mathrm{seed}}
 &\coloneqq\bigl(\ket{\omega(\hat{\alpha})}\bra{\omega(\hat{\alpha})}
     +\delta_{\mathrm{seed}}\Pi_{\mathcal W}\bigr)
        \oplus B_{\mathrm{iso}}(\hat{\alpha}),\label{eq:slow-family-raw-seed}\\
 R_{\mathrm{seed}}&\coloneqq\tr_A\overline X_{\mathrm{seed}},\label{eq:slow-family-seed-marginal}\\
 X^{(0)}&\coloneqq(I_A\otimes R_{\mathrm{seed}}^{-1/2})
       \overline X_{\mathrm{seed}}(I_A\otimes R_{\mathrm{seed}}^{-1/2}).
 \label{eq:slow-family-normalized-seed}
\end{align}
\end{subequations}
The marginal $R_{\mathrm{seed}}$ is diagonal, with entries $1+\delta_{\mathrm{seed}}$ at $b_0,b_1,b_n$ and $1+2\delta_{\mathrm{seed}}$ at the remaining basis vectors. The coupled block is positive definite, as is the isolated block because its determinant is $(1-\chi_{\mathrm{corr}}^2)\cos^2\hat{\alpha}_1\sin^2\hat{\alpha}_n>0$. Diagonal normalization preserves these supports. Hence $\tr_A X^{(0)}=I_B$ and $\supp X^{(0)}=\supp C$, which verifies Eq.~\eqref{eq:support-condition}. Every finite iterate remains positive definite on this support.

The iteration preserves $\mathcal W\oplus\mathcal U$ and has a diagonal normalizer. To describe its limiting scalar dynamics, first restrict it to $X_{\mathrm{br}}(\alpha)$. The transformed coordinates in each circle block are
\begin{subequations}\label{eq:slow-branch-coordinates}
\begin{align}
 z_0&=1-\gamma_n\sum_{j=1}^{n-1}\ell_j,\label{eq:slow-branch-anchor}\\
 b_{1,x}&=\cos\alpha_1,\label{eq:slow-branch-first-x}\\
 b_{j,x}&=\cos\alpha_j+\gamma_n\ell_{j-1}\quad(2\le j\le n),\label{eq:slow-branch-x}\\
 b_{1,y}&=\sin\alpha_1-\gamma_n(\sin\alpha_1-\ell_1),\label{eq:slow-branch-first-y}\\
 b_{j,y}&=\sin\alpha_j+\gamma_n\ell_j\quad(2\le j<n),\label{eq:slow-branch-y}\\
 b_{n,y}&=\sin\alpha_n.\label{eq:slow-branch-last-y}
\end{align}
\end{subequations}
Here $b_{j,x},b_{j,y}$ are scalar coordinates. In a neighborhood of zero, $z_0$ and every $b_{j,x}$ are positive. The branch normalizer $Y(\alpha)$ has entries $Y_0(\alpha)=z_0$ and $Y_j(\alpha)=(b_{j,x}^2+b_{j,y}^2)^{1/2}$. The updated angles are $\mathcal T_j(\alpha)\coloneqq\arctan(b_{j,y}/b_{j,x})$, and the RW map sends $X_{\mathrm{br}}(\alpha)$ to $X_{\mathrm{br}}(\mathcal T(\alpha))$. These branch formulas will approximate the actual iterates once their positive covariance has been controlled.

Let $d_j(\alpha)\coloneqq b_{j,x}\cos\alpha_j+b_{j,y}\sin\alpha_j$. Taking the component tangent to the circle gives the exact relation
\begin{equation}\label{eq:slow-angle-map}
 \tan(\mathcal T_j(\alpha)-\alpha_j)
 =-\frac{\gamma_n\partial_{\alpha_j}g(\alpha)}{2d_j(\alpha)},
\end{equation}
where $d_j(0)=1$. Thus the first $n-1$ angles are fixed on $\bar\alpha(t)$, while Eqs.~\eqref{eq:slow-family-stationary-gap} and \eqref{eq:slow-angle-map} imply
\begin{equation}\label{eq:slow-branch-drift}
 \mathcal T_n(\bar\alpha(t))-t
 =-m_nc_nt^{2m_n-1}(1+o(1)).
\end{equation}
At zero, the Jacobian of the first $n-1$ components with respect to $u$ is $\operatorname{diag}(1-2\gamma_n,1-\gamma_n,\ldots,1-\gamma_n)$, so these components contract locally for fixed $t$. Also $\partial_{\alpha_j}\mathcal T_n=\OCal(t)$ for $j<n$, which suppresses the effect of deviations from the stationary curve on the terminal angle.

\subsection{Control of the positive covariance}
\label{sec:slow-covariance}

To transfer the scalar dynamics to the chosen seed, we control the positive covariance that ensures sufficient support. Throughout this subsection, write $m\coloneqq m_n$. Constants may depend on fixed $n$ but are uniform in the iteration index and the coordinates in the neighborhoods specified below. Quantities without an iteration superscript refer to the current step. All iterates are real because the cost matrix and seed are real.

Normalization fixes the anchor diagonal entry at one, so the coupled block has the unique representation
\begin{equation}\label{eq:slow-cov-coordinates}
 X_{\mathcal W}=
 \begin{pmatrix}1&x^{\mathsf T}\\x&xx^{\mathsf T}+\Sigma_{\mathrm{cov}}\end{pmatrix},
\end{equation}
where $\Sigma_{\mathrm{cov}}\succ0$ and the entries of $x$ are ordered as $q_1,p_2,\ldots,q_{n-1},p_n$. Partition the fixed cost matrix as $C_{\mathcal W}=\left(\begin{smallmatrix}c_{00}&b_{\mathcal W}^{\mathsf T}\\b_{\mathcal W}&K_{\mathcal W}\end{smallmatrix}\right)$ and define
\begin{subequations}\label{eq:slow-cov-schur-data}
 \begin{align}
 p_{\mathrm{anc}}&\coloneqq c_{00}+b_{\mathcal W}^{\mathsf T}x,
 \label{eq:slow-cov-anchor}\\
 v_{\mathrm{tail}}&\coloneqq b_{\mathcal W}+K_{\mathcal W}x,
 \label{eq:slow-cov-tail}\\
 F_{\mathrm{cov}}&\coloneqq K_{\mathcal W}
 -\frac{v_{\mathrm{tail}}b_{\mathcal W}^{\mathsf T}}{p_{\mathrm{anc}}},
 \label{eq:slow-cov-schur-factor}\\
 \sigma_{\mathrm{anc}}&\coloneqq p_{\mathrm{anc}}^2+
 b_{\mathcal W}^{\mathsf T}\Sigma_{\mathrm{cov}}b_{\mathcal W}.
 \label{eq:slow-cov-anchor-variance}
 \end{align}
\end{subequations}
Here $p_{\mathrm{anc}}$ is positive near the limiting branch. Let $D_{\mathrm{tail}}$ be the diagonal restriction of $I_A\otimes Y^{(k)}$ to the nonanchor coordinates of $\mathcal W$. Taking the Schur complement after the congruence update gives the following covariance update and upper bound
\begin{equation}\label{eq:slow-cov-update}
 \begin{aligned}
 \Sigma_{\mathrm{cov}}^{(k+1)}
 &=T_{\mathrm{cov}}\left(\Sigma_{\mathrm{cov}}
 -\frac{\Sigma_{\mathrm{cov}}b_{\mathcal W}b_{\mathcal W}^{\mathsf T}
 \Sigma_{\mathrm{cov}}}{\sigma_{\mathrm{anc}}}\right)T_{\mathrm{cov}}^{\mathsf T},\\
 T_{\mathrm{cov}}&\coloneqq D_{\mathrm{tail}}^{-1}F_{\mathrm{cov}},\qquad
 0\preceq\Sigma_{\mathrm{cov}}^{(k+1)}
 \preceq T_{\mathrm{cov}}\Sigma_{\mathrm{cov}}T_{\mathrm{cov}}^{\mathsf T}.
 \end{aligned}
\end{equation}
Indeed, the middle matrix equals $(\Sigma_{\mathrm{cov}}^{-1}+b_{\mathcal W}b_{\mathcal W}^{\mathsf T}/p_{\mathrm{anc}}^2)^{-1}$, which is positive and bounded above by $\Sigma_{\mathrm{cov}}$.

For the actual iterates, define $x_{q_j}=\sin\alpha_j$ for $j<n$ and write the isolated $q_n$ population as $b=\sin^2t$. If $a$ denotes the isolated $p_1$ population, normalization determines the remaining positive coordinates by
\begin{equation}\label{eq:slow-cov-normalization}
 \begin{aligned}
 a&=\cos^2\alpha_1-[\Sigma_{\mathrm{cov}}]_{q_1q_1},\\
 x_{p_j}^2&=\cos^2\alpha_j-[\Sigma_{\mathrm{cov}}]_{p_jp_j}
 -[\Sigma_{\mathrm{cov}}]_{q_jq_j}\quad(2\le j<n),\\
 x_{p_n}^2&=\cos^2t-[\Sigma_{\mathrm{cov}}]_{p_np_n}.
 \end{aligned}
\end{equation}
The positive square roots determine $x$. Direct multiplication gives $x^{(k+1)}=D_{\mathrm{tail}}^{-1}(p_{\mathrm{anc}}v_{\mathrm{tail}}+K_{\mathcal W}\Sigma_{\mathrm{cov}}b_{\mathcal W})/\sqrt{\sigma_{\mathrm{anc}}}$. The updated angles are $\alpha_j^{(k+1)}=\arcsin x_{q_j}^{(k+1)}$ and $t^{(k+1)}=\arcsin(\sin t/Y_n^{(k)})$. Choose a closed ball $\norm{u}_2+|t|+\opnorm{\Sigma_{\mathrm{cov}}}\le\rho$ inside an open neighborhood where all radicands and $p_{\mathrm{anc}},Y_j,d_j$ exceed $1/2$ and the inverse sine arguments have absolute value below $1/2$. These formulas extend analytically through $t=0$, with bounded derivatives on the ball. They agree with the branch map at zero covariance. Writing $\mathcal T_u\coloneqq(\mathcal T_1,\ldots,\mathcal T_{n-1})$, the mean value theorem gives
\begin{equation}\label{eq:slow-cov-angle-error}
 \begin{aligned}
 u^{(k+1)}&=\mathcal T_u(u,t)+\OCal(\opnorm{\Sigma_{\mathrm{cov}}}),\\
 t^{(k+1)}&=\mathcal T_n(u,t)+\OCal(t\opnorm{\Sigma_{\mathrm{cov}}}).
 \end{aligned}
\end{equation}
The covariance derivative of the terminal formula contains the factor $\sin t$, giving the additional factor $t$ in the second estimate.

To estimate covariance contraction, first evaluate the branch normalizers on $\bar\alpha(t)=(h_n(t),t)$. Abbreviate $a_j\coloneqq Y_j(\bar\alpha(t))-1$ and $\ell_j\coloneqq\ell_j(\bar\alpha(t))$. Fast stationarity gives the following relations, where we also define $\widehat a_j$ for later use.
\begin{equation}\label{eq:slow-cov-stationary-normalizers}
 \begin{aligned}
 a_0&=-\gamma_n\sum_{j=1}^{n-1}\ell_j,\qquad a_1=0,\\
 \widehat a_j&\coloneqq\frac{\gamma_n\ell_{j-1}}{\cos\bar\alpha_j}
 \quad(2\le j\le n),\qquad a_j=\widehat a_j\quad(2\le j<n).
 \end{aligned}
\end{equation}
Since the derivative of the reduced gap is $2\gamma_n\ell_{n-1}\sin t$, its leading term implies $\gamma_n\ell_{n-1}\sim mc_nt^{2m-2}$. The terminal branch vector is $(\cos t+\gamma_n\ell_{n-1},\sin t)$, which yields the positive correction
\begin{equation}\label{eq:slow-cov-terminal-correction}
 \begin{aligned}
 \delta_n(t)&\coloneqq\widehat a_n-a_n\\
 &=\frac{\gamma_n\ell_{n-1}\sin^2t}{\cos t}
 +\OCal(\ell_{n-1}^2\sin^2t)
 \sim mc_nt^{2m}.
 \end{aligned}
\end{equation}
Let $D_{\mathrm{core}}(t)$ be the restriction of $I_A\otimes Y(\bar\alpha(t))$ to $\mathcal W$. Replacing its $p_n$ entry by $1+\widehat a_n$ defines $\widehat D_{\mathrm{core}}(t)$ and the comparison slack
\begin{equation}\label{eq:slow-comparison-slack}
 \begin{aligned}
 \Delta_{\mathrm{cmp}}(t)
 &\coloneqq\widehat D_{\mathrm{core}}(t)-C_{\mathcal W}\\
 &=D_{\mathrm{core}}(t)-C_{\mathcal W}
 +\delta_n(t)\ket{p_n}\bra{p_n}.
 \end{aligned}
\end{equation}
Its positivity follows from completed squares. For $1\le j<n$, we use the following pair matrices and vectors, which satisfy the displayed identity by stationarity.
\begin{equation}\label{eq:slow-cov-pair-blocks}
 \begin{aligned}
 B_j&\coloneqq\begin{pmatrix}
 \gamma_n(1+\delta_{j1})+a_j&\gamma_n\\
 \gamma_n&\gamma_n+\widehat a_{j+1}
 \end{pmatrix},\\
 \upsilon_j&\coloneqq\begin{pmatrix}\sin\bar\alpha_j\\\cos\bar\alpha_{j+1}\end{pmatrix},
 \qquad B_j\upsilon_j=\gamma_n\begin{pmatrix}1\\1\end{pmatrix},
 \end{aligned}
\end{equation}
where $\delta_{j1}$ is the Kronecker delta. A vector $\zeta\in\mathcal W$, with anchor coordinate $\zeta_0$ and pair coordinates $\zeta_j$, then satisfies
\begin{equation}\label{eq:slow-cov-completed-squares}
 \bra{\zeta}\Delta_{\mathrm{cmp}}(t)\ket{\zeta}
 =\sum_{j=1}^{n-1}(\zeta_j-\zeta_0\upsilon_j)^\dagger
 B_j(\zeta_j-\zeta_0\upsilon_j).
\end{equation}
The first pair matrix has a positive definite limit. For $j\ge2$, its determinant is $\gamma_n(\widehat a_j+\widehat a_{j+1})+\widehat a_j\widehat a_{j+1}>0$. The stationary relations give $\ell_j=\Theta(t^{2m-2^{n-j}})$, so every pair has least eigenvalue at least a positive multiple of $t^{2m-4}$. The pair traces remain bounded. Since the distance to the kernel is bounded by $(\sum_j\norm{\zeta_j-\zeta_0\upsilon_j}_2^2)^{1/2}$, Eq.~\eqref{eq:slow-cov-completed-squares} gives
\begin{equation}\label{eq:slow-comparison-gap}
 \begin{aligned}
 \Delta_{\mathrm{cmp}}(t)&\succeq0,\qquad
 \ker\Delta_{\mathrm{cmp}}(t)=\operatorname{span}\{\ket{\omega(\bar\alpha(t))}\},\\
 \lambda_{\mathrm{gap}}(t)&\ge c_{\mathrm{gap}}t^{2m-4},\qquad c_{\mathrm{gap}}>0.
 \end{aligned}
\end{equation}
Here $\lambda_{\mathrm{gap}}$ is the least positive eigenvalue. When $n=2$, only $B_1$ occurs and the bound is uniform.

This gap controls the covariance in a metric that stays regular as $t\to0$. Let $\omega_{\mathrm{act}}(t)\coloneqq(1,x_{\mathrm{stat}}(t)^{\mathsf T})^{\mathsf T}$ be the coordinate vector of $\ket{\omega(\bar\alpha(t))}$, and put
\begin{equation}\label{eq:slow-cov-metric}
 \begin{aligned}
 \Pi_t&\coloneqq I-\frac{\omega_{\mathrm{act}}\omega_{\mathrm{act}}^{\mathsf T}
 \widehat D_{\mathrm{core}}}
 {\omega_{\mathrm{act}}^{\mathsf T}\widehat D_{\mathrm{core}}\omega_{\mathrm{act}}},
 \qquad E\coloneqq\begin{pmatrix}0\\I\end{pmatrix},\\
 M_{\mathrm{met}}(t)&\coloneqq E^{\mathsf T}\Pi_t^{\mathsf T}
 \widehat D_{\mathrm{core}}(t)\Pi_t E.
 \end{aligned}
\end{equation}
The projection $\Pi_t$ removes the kernel direction in the $\widehat D_{\mathrm{core}}$ metric. At zero, $M_{\mathrm{met}}(0)=I-x_{\mathrm{stat}}(0)x_{\mathrm{stat}}(0)^{\mathsf T}/n$ and $\norm{x_{\mathrm{stat}}(0)}_2^2=n-1$. Thus, after shrinking the interval, $(2n)^{-1}I\preceq M_{\mathrm{met}}(t)\preceq2I$, and its derivative is bounded by analyticity.

The operator $\widehat A\coloneqq\widehat D_{\mathrm{core}}^{-1}C_{\mathcal W}$ satisfies $\widehat A^{\mathsf T}\widehat D_{\mathrm{core}}=\widehat D_{\mathrm{core}}\widehat A=C_{\mathcal W}$ and fixes $\omega_{\mathrm{act}}$. Its self-adjointness therefore gives $\Pi_t\widehat A=\widehat A\Pi_t$. Since $I-\widehat D_{\mathrm{core}}^{-1/2}C_{\mathcal W}\widehat D_{\mathrm{core}}^{-1/2}=\widehat D_{\mathrm{core}}^{-1/2}\Delta_{\mathrm{cmp}}\widehat D_{\mathrm{core}}^{-1/2}$, Eq.~\eqref{eq:slow-comparison-gap} and uniform bounds on $\widehat D_{\mathrm{core}}$ place the remaining eigenvalues in $(0,1-ct^{2m-4}]$. Let $\widehat T_{\mathrm{cov}}$ be $T_{\mathrm{cov}}$ evaluated at $x_{\mathrm{stat}}(t)$ with the tail of $\widehat D_{\mathrm{core}}$ as normalizer. Here $p_{\mathrm{anc}}=(\widehat D_{\mathrm{core}})_{00}$ and $v_{\mathrm{tail}}=\widehat D_{\mathrm{tail}}x_{\mathrm{stat}}$, where $\widehat D_{\mathrm{tail}}$ denotes that tail. Substitution in $F_{\mathrm{cov}}$ gives the first identity below. Applying $\Pi_t$ and the spectral bound gives the second.
\begin{equation}\label{eq:slow-cov-quotient}
 \begin{aligned}
 E\widehat T_{\mathrm{cov}}\zeta
 &=\widehat A E\zeta-\omega_{\mathrm{act}}(\widehat A E\zeta)_0,\\
 \widehat T_{\mathrm{cov}}^{\mathsf T}M_{\mathrm{met}}
 \widehat T_{\mathrm{cov}}
 &\preceq(1-c_{\mathrm{ctr}}t^{2m-4})M_{\mathrm{met}},
 \end{aligned}
\end{equation}
Here $c_{\mathrm{ctr}}>0$ is independent of $t$.

For an actual iterate, define the fast error $e\coloneqq\norm{u-h_n(t)}_2$ and the weighted covariance size $h_{\mathrm{cov}}\coloneqq\tr(M_{\mathrm{met}}(t)\Sigma_{\mathrm{cov}})$. For constants $K>0$ and $t_*>0$ to be chosen below, consider the region
\begin{equation}\label{eq:slow-cov-neighborhood}
 0<t\le t_*,\qquad e\le Kt^{2m},\qquad h_{\mathrm{cov}}\le t^{2m+2}.
\end{equation}
These bounds keep the iterates close to the stationary curve while controlling their covariance. We first derive estimates within this region, then verify that the initial seed lies in it and that the iteration preserves it. The metric bounds give $h_{\mathrm{cov}}\asymp\opnorm{\Sigma_{\mathrm{cov}}}$. The uniform derivative bounds and $\widehat D_{\mathrm{core}}-D_{\mathrm{core}}=\OCal(t^{2m})$ imply
\begin{equation}\label{eq:slow-cov-perturbation}
 \begin{aligned}
 \opnorm{T_{\mathrm{cov}}-\widehat T_{\mathrm{cov}}}
 &=\OCal(e+h_{\mathrm{cov}}+t^{2m}),\\
 \opnorm{M_{\mathrm{met}}(t^{(k+1)})-M_{\mathrm{met}}(t)}
 &=\OCal(|t^{(k+1)}-t|).
 \end{aligned}
\end{equation}
The stationary scalar law and $\partial_{\alpha_j}\mathcal T_n=\OCal(t)$ first give $|t^{(k+1)}-t|=\OCal(t^{2m-1}+te+th_{\mathrm{cov}})$ with a constant independent of $K$. Under Eq.~\eqref{eq:slow-cov-neighborhood}, this becomes $t^{(k+1)}=t-mc_nt^{2m-1}(1+o(1))$, uniformly for fixed $K$. Consequently the perturbation of the covariance contraction is $\OCal(t^{2m-1})$, smaller than its margin $t^{2m-4}$ by $\OCal(t^3)$. Fix $q_{\mathrm{ctr}}\in(1-\gamma_n,1)$. The limiting fast Jacobian has norm at most $1-\gamma_n$, so continuity and $h_n'(t)=\OCal(t)$ give, on a smaller interval,
\begin{equation}\label{eq:slow-cov-bootstrap}
 \begin{aligned}
 h_{\mathrm{cov}}^{(k+1)}&\le(1-c_*t^{2m-4})h_{\mathrm{cov}},\\
 e^{(k+1)}&\le q_{\mathrm{ctr}}e
 +C_*t^{2m}+C_*h_{\mathrm{cov}},
 \end{aligned}
\end{equation}
Here one may take $c_*=c_{\mathrm{ctr}}/2$. The forcing $C_*t^{2m}$ comes from the curve displacement $\norm{h_n(t^{(k+1)})-h_n(t)}_2=\OCal(t|t^{(k+1)}-t|)$. Absorbing its additional $\OCal(t^2e)$ term into $q_{\mathrm{ctr}}$ fixes $C_*>0$ independently of $K$. Choose $K>4C_* /(1-q_{\mathrm{ctr}})$ and then restrict $0<t\le t_*$. For sufficiently small $t_*$, the scalar decrement lies between $(mc_n/2)t^{2m-1}$ and $2mc_nt^{2m-1}$, with $0<t^{(k+1)}<t$. The fast error is at most $[K-(1-q_{\mathrm{ctr}})K/2]t^{2m}$, whereas $K(t^{(k+1)})^{2m}=Kt^{2m}[1-\OCal(t^{2m-2})]$. Also $(t^{(k+1)}/t)^{2m+2}=1-\OCal(t^{2m-2})$, whose loss is smaller than $c_*t^{2m-4}$. Shrinking $t_*$ once more therefore preserves the region in Eq.~\eqref{eq:slow-cov-neighborhood}.

The normalized seed lies in this region. To check this directly, let $D_{\mathrm{seed}}$ be the diagonal square root of its marginal restricted to the core tail, and let $\omega_{\mathrm{seed}}$ denote the tail of the raw ket $\omega(\bar\alpha(t_0))$. Its Schur covariance is
\begin{equation}\label{eq:slow-cov-seed}
 \Sigma_{\mathrm{cov}}^{(0)}
 =\delta D_{\mathrm{seed}}^{-2}
 +\frac{\delta}{1+\delta}D_{\mathrm{seed}}^{-1}
 \omega_{\mathrm{seed}}\omega_{\mathrm{seed}}^{\mathsf T}
 D_{\mathrm{seed}}^{-1},
\end{equation}
where $\delta\coloneqq t_0^{4m}$. Normalization changes the fast coordinates by $\OCal(\delta)$ and the terminal angle by $\OCal(\delta t_0)$. Thus $t^{(0)}\sim t_0$, $e^{(0)}=\OCal(t_0^{4m})$, and $h_{\mathrm{cov}}^{(0)}=\OCal(t_0^{4m})$. Since $4m>2m+2$ for $m\ge2$, the seed satisfies the invariant bounds for every sufficiently small $t_0>0$.

The terminal angle remains positive, decreases, and tends to zero. Inverting its scalar decrement gives
\begin{equation}\label{eq:slow-family-trajectory-rate}
 \begin{aligned}
 t^{(k+1)}&=t^{(k)}-m_nc_n(t^{(k)})^{2m_n-1}(1+o(1)),\\
 t^{(k)}&\sim(c_{\mathrm{time},n}k)^{-1/(2m_n-2)},\\
 c_{\mathrm{time},n}&\coloneqq2m_n(m_n-1)c_n.
 \end{aligned}
\end{equation}
Summing the covariance contraction in Eq.~\eqref{eq:slow-cov-bootstrap} now gives
\begin{equation}\label{eq:slow-cov-decay}
 h_{\mathrm{cov}}^{(k)}\le C
 \begin{cases}
 \exp(-ck),&n=2,\\
 \exp(-ck^{1/(m_n-1)}),&n>2,
 \end{cases}
 \qquad c,C>0.
\end{equation}
In particular, $\opnorm{\Sigma_{\mathrm{cov}}^{(k)}}=o((t^{(k)})^q)$ for every fixed $q>0$. The covariance remains positive at every finite step, but it does not affect the algebraic orders of the slow iterative update.

\subsection{Objective, certificate, and matrix rates}
\label{sec:circle-certificate-matrix}

The scalar asymptotic and covariance estimate now determine the three accuracy measures. For a normalized iterate supported on $P$, the objective gap is exactly $r^{(k)}=\gamma_n\sum_{j=0}^{n-1}\bra{w_j}X^{(k)}\ket{w_j}$. Its difference from the boundary gap at the same angles is $\OCal(\opnorm{\Sigma_{\mathrm{cov}}^{(k)}})$, while stationarity makes the error from the fast coordinates quadratic in $e^{(k)}$. Thus Eqs.~\eqref{eq:slow-family-stationary-gap} and \eqref{eq:slow-family-trajectory-rate} give
\begin{equation}\label{eq:slow-family-objective-rate}
 r^{(k)}\sim c_n(t^{(k)})^{2m_n}
 \sim c_n(c_{\mathrm{time},n}k)^{-m_n/(m_n-1)}.
\end{equation}

For the certificate, let $\Delta^{(k)}\coloneqq I_A\otimes Y^{(k)}-C$. Since $s=0$ and $Y^{(k)}\succ0$, Eq.~\eqref{eq:certificate} reduces to
\begin{equation}\label{eq:circle-certificate-reduction}
\begin{aligned}
 \tau^{(k)}&=\max\{0,-\lambda_{\min}(\Delta^{(k)})\},\\
 r_{\mathrm{cert}}^{(k)}&=\tr Y^{(k)}-f^{(k)}+(n+1)\tau^{(k)}.
\end{aligned}
\end{equation}
The comparison slack in Eq.~\eqref{eq:slow-comparison-slack} identifies the leading eigenvalue. Write $t\coloneqq t^{(k)}$ and define the perturbation on $\mathcal W$ by
\begin{equation}\label{eq:slow-normalizer-perturbation}
E_{\mathcal W}^{(k)}\coloneqq
 \Pi_{\mathcal W}(I_A\otimes[Y^{(k)}-Y(\bar\alpha(t))])\Pi_{\mathcal W}.
\end{equation}
Its norm is $\OCal(t^{2m_n})$, which is the order of the eigenvalue sought, so a norm bound alone does not suffice. The needed cancellation follows from $F_0(\alpha)\coloneqq Y_0(\alpha)+\sum_{j=2}^nY_j(\alpha)$. In Eq.~\eqref{eq:slow-branch-coordinates}, the linear terms in $Y_0$ and $Y_j$ are $\gamma_n\sum_{j<n}\alpha_j$ and $-\gamma_n\alpha_{j-1}$, respectively. Hence $\partial_\alpha F_0(0)=0$.

Put $\ket{\omega_0}\coloneqq\ket{p_0}+\sum_{j=2}^n\ket{p_j}$ and $\alpha^{(k)}\coloneqq(u^{(k)},t)$. Bounded second derivatives and $\partial_\alpha F_0(0)=0$ bound the derivative along the segment from $\bar\alpha(t)$ to $\alpha^{(k)}$ by $\OCal(t+e^{(k)})$. The mean value theorem therefore gives
\begin{equation}\label{eq:slow-kernel-cancellation}
\begin{aligned}
 \bra{\omega_0}E_{\mathcal W}^{(k)}\ket{\omega_0}
 &=F_0(\alpha^{(k)})-F_0(\bar\alpha(t))
   +\OCal(\opnorm{\Sigma_{\mathrm{cov}}^{(k)}})\\
 &=\OCal\!\left(te^{(k)}+(e^{(k)})^2
           +\opnorm{\Sigma_{\mathrm{cov}}^{(k)}}\right)\\
 &=o(t^{2m_n}).
\end{aligned}
\end{equation}
Replacing $\ket{\omega_0}$ by $\ket{\omega(\bar\alpha(t))}=\ket{\omega_0}+\OCal(t^2)$ changes the compression by $\OCal(t^2\opnorm{E_{\mathcal W}^{(k)}})$, preserving this estimate. On $\mathcal W$, write the actual slack as $\Delta_{\mathrm{cmp}}+R_{\mathrm{pert}}$, where $R_{\mathrm{pert}}\coloneqq E_{\mathcal W}^{(k)}-\delta_n(t)\ket{p_n}\bra{p_n}$. Its norm divided by $\lambda_{\mathrm{gap}}$ is $\OCal(t^4)$. Weyl's inequality isolates one eigenvalue $\lambda$ near zero with $|\lambda|\le\opnorm{R_{\mathrm{pert}}}$, while the complementary block minus $\lambda I$ is bounded below by $\lambda_{\mathrm{gap}}I/2$. For the unit kernel vector $\widehat\omega\coloneqq\ket{\omega(\bar\alpha(t))}/\norm{\ket{\omega(\bar\alpha(t))}}_2$, its Schur complement gives
\begin{equation}\label{eq:slow-certificate-eigenvalue}
\begin{aligned}
 \lambda
 &=\inner{\widehat\omega}{R_{\mathrm{pert}}\widehat\omega}
   +\OCal\!\left(\frac{\opnorm{R_{\mathrm{pert}}}^2}{\lambda_{\mathrm{gap}}}\right)\\
 &=-\frac{\delta_n(t)\cos^2t}
         {\norm{\ket{\omega(\bar\alpha(t))}}_2^2}+o(t^{2m_n})\\
 &\sim-\frac{m_nc_n}{n}t^{2m_n}.
\end{aligned}
\end{equation}
The remainder is $\OCal(t^{2m_n+4})$ and the denominator tends to $n$. In particular, $\lambda$ is the least eigenvalue of the actual slack on $\mathcal W$.

The remaining blocks do not alter this coefficient. At $p_1$, the slack eigenvalue is $Y_1^{(k)}-1=o(t^{2m_n})$, because $Y_1(\bar\alpha(t))=1$ and $\partial_\alpha Y_1(0)=0$. At $q_n$, it is asymptotic to $m_nc_nt^{2m_n-2}>0$. Outside $\supp P$, the cost matrix vanishes and the slack is positive. Thus $\tau^{(k)}\sim(m_nc_n/n)t^{2m_n}$.

The uncorrected trace discrepancy is of smaller order. On the boundary family, taking the inner product of each unit circle vector with its transformed vector gives
\begin{equation}\label{eq:slow-trace-discrepancy}
\begin{aligned}
 &\tr Y(\alpha)-\tr(CX_{\mathrm{br}}(\alpha))\\
 &\qquad=\sum_{j=1}^nY_j(\alpha)
       [1-\cos(\mathcal T_j(\alpha)-\alpha_j)].
\end{aligned}
\end{equation}
The fast angle changes are $\OCal(e^{(k)})$, and the terminal change is $\OCal(t^{2m_n-1})$. The covariance contributes less than every fixed power of $t$, so the actual trace discrepancy is $\OCal(t^{4m_n-2})=o(t^{2m_n})$. Substitution in Eq.~\eqref{eq:circle-certificate-reduction} proves
\begin{equation}\label{eq:circle-certificate-rate}
 r_{\mathrm{cert}}^{(k)}
 \sim\frac{(n+1)m_nc_n}{n}(t^{(k)})^{2m_n}
 \sim\frac{(n+1)m_n}{n}\,r^{(k)}.
\end{equation}

To determine the matrix rate, write the isolated block of $X^{(k)}$ as $\left(\begin{smallmatrix}a^{(k)}&z^{(k)}\\z^{(k)}&b^{(k)}\end{smallmatrix}\right)$. Its entries are divided by $(Y_1^{(k)})^2$, $Y_1^{(k)}Y_n^{(k)}$, and $(Y_n^{(k)})^2$, respectively, so $z^{(k)}/\sqrt{a^{(k)}b^{(k)}}=\chi_{\mathrm{corr}}$ at every step. All fast angles are $\OCal((t^{(k)})^2)$, while $b^{(k)}=\sin^2t^{(k)}$. It follows that
\begin{equation}\label{eq:slow-matrix-limit}
 \lim_{k\to\infty}X^{(k)}=X^{(\infty)}
 =\ket{\omega_0}\bra{\omega_0}+\ket{p_1}\bra{p_1}.
\end{equation}
The coupled block and isolated diagonal entries differ from their limits by $\OCal((t^{(k)})^2)$. The leading error instead comes from $z^{(k)}\sim\chi_{\mathrm{corr}}t^{(k)}$, giving
\begin{equation}\label{eq:circle-matrix-rate}
 \fnorm{X^{(k)}-X^{(\infty)}}
 \sim\sqrt2\,\chi_{\mathrm{corr}}t^{(k)}.
\end{equation}
The nonzero isolated correlation is therefore essential for this matrix exponent.

Define $k_{\mathrm{obj},n}(\varepsilon)$, $k_{\mathrm{cert},n}(\varepsilon)$, and $k_{X,n}(\varepsilon)$ as the first indices $k\ge0$ at which the respective errors are at most $\varepsilon$. Inverting the preceding equivalents yields
\begin{equation}\label{eq:slow-family-iteration-count}
\begin{aligned}
 k_{\mathrm{obj},n}(\varepsilon)
 &\sim\frac{c_n^{-1/m_n}}{2m_n(m_n-1)}
     \varepsilon^{-(m_n-1)/m_n},\\
 k_{\mathrm{cert},n}(\varepsilon)
 &\sim\left(\frac{(n+1)m_n}{n}\right)^{(m_n-1)/m_n}
     k_{\mathrm{obj},n}(\varepsilon),\\
 k_{X,n}(\varepsilon)
 &\sim\frac{(\sqrt2\chi_{\mathrm{corr}})^{2m_n-2}}{c_{\mathrm{time},n}}
     \varepsilon^{-(2m_n-2)}.
\end{aligned}
\end{equation}
Every leading constant in Eq.~\eqref{eq:slow-family-iteration-count} is positive, so these equivalents establish the $\Theta$ estimates in Eq.~\eqref{eq:circle-exponent-comparison}. Monotonicity of the certificate or matrix error is not required. Each error is positive at every finite iterate, and its asymptotic equivalent bounds all sufficiently late iterates from both sides. A finite initial segment therefore cannot affect the first hitting asymptotics.

\section{Numerical simulation details}
\label{sec:numerics}
This appendix describes the sampling, initialization, and accuracy conventions for Figures~\ref{fig:timing} and~\ref{fig:gue}.

\subsection{Runtime comparison}
\label{sec:runtime-details}
Figure~\ref{fig:timing} compares the original RW and JRF iterations with \texttt{CVXPY}~\cite{DiamondBoyd2016} models solved by the first order solver SCS~\cite{ODonoghue2016} and the interior point solver Clarabel~\cite{GoulartChen2026}. We fix $d_A=4$ and use $500$ independent instances for each iterative method at each sampled dimension, with the SCS and Clarabel methods evaluated on the same subset of $20$ instances. All methods use an objective error target of $10^{-7}$. For discrimination, we draw equiprobable states independently from the Hilbert--Schmidt measure and initialize the JRF iteration at $M_i^{(0)}=I_B/d_A$. For entanglement fidelity, the input is $\rho=I_A/d_A$ and the noise channel is induced by a Haar random isometry from $\Cbb^{d_A}$ to $\Cbb^{d_B}\otimes\Cbb^{d_B}$, with the second factor traced out. The RW iteration starts at $X^{(0)}=I_{AB}/d_A$. Its cost matrix is singular, while $\tr_A C\succ0$ holds for every instance of both families.

All methods meet the objective accuracy target using a precomputed numerical dual upper bound whose gap from a primal value is at most $2\times10^{-11}$, including roundoff allowances. The outputs from the SCS and Clarabel methods are corrected for positivity and normalization before their objectives are evaluated. Timings exclude instance generation and reference construction. For the RW and JRF iterations, they include initialization, updates, and objective evaluation, but not evaluation of their own certificate gaps. For the \texttt{CVXPY} implementations using the SCS and Clarabel methods, the reported runtimes include model construction, compilation, solving, feasibility correction, and cumulative retries. Runs use one numerical-library thread per worker, with each worker assigned to a single processor core.

\subsection{GUE iteration counts}
\label{sec:gue-details}
For Figure~\ref{fig:gue}, we draw $500$ independent GUE cost matrices at each dimension $d_A=d_B=d\in\{4,6,8,10\}$, with the Wigner normalization specified in Section~\ref{sec:further-applications}. Every sampled cost matrix has negative eigenvalues. We use the minimal admissible shift $s=-\lambda_{\min}(C)$ and the positive definite seed pair $X^{(0)}=I_{AB}$ and $S^{(0)}=I_B$, which satisfies Eq.~\eqref{eq:support-condition}.

For each run of the iteration, let $k_{\mathrm{stop}}$ be the first index $k\ge2$ for which $f^{(k)}-f^{(k-1)}\le10^{-14}$. We record the objective value at this stopping point as $f^{\mathrm{stop}}\coloneqq f^{(k_{\mathrm{stop}})}=\tr(CX^{(k_{\mathrm{stop}})})$. To obtain a refined numerical matrix reference, we continue the same iteration for at least $1000$ additional steps until the certificate gap of Eq.~\eqref{eq:certificate-gap} is at most $10^{-10}$ and the Frobenius distance between iterates $250$ steps apart is at most $10^{-8}$. Every instance reaches this later reference iterate $X^{\mathrm{ref}}$, whose objective lies within $10^{-10}$ of $\fopt$, up to roundoff.

For each tolerance $\varepsilon$, we record the first iteration $k(\varepsilon)\ge1$ at which the chosen error measure is at most $\varepsilon$. The first row of Figure~\ref{fig:gue} measures the objective discrepancy as $\max\{0,f^{\mathrm{stop}}-f^{(k)}\}$. The second row uses the certificate gap $r_{\mathrm{cert}}^{(k)}$, with a roundoff allowance added to the dual bound. The third row measures the empirical distance $\fnorm{X^{(k)}-X^{\mathrm{ref}}}$ to the refined numerical reference.

For each dimension, let $f_i^{\mathrm{upbnd}}$ denote the smallest recorded numerical dual upper bound for instance $i$, including a roundoff allowance. We estimate the uncertainty in its stopping value by $\delta_i\coloneqq\max\{0,f_i^{\mathrm{upbnd}}-f_i^{\mathrm{stop}}\}$. Among the tested objective tolerances, we display only those satisfying
\begin{equation}\label{eq:displayed-objective-tolerance}
 \varepsilon\ge10\max_{1\le i\le500}\delta_i.
\end{equation}
This ensures that the estimated error in each stopping value is at most one tenth of the displayed tolerance. Every instance must also reach each displayed target, and tolerances below the numerical precision floor are excluded.

For each dimension, let $d_i$ denote the Frobenius distance between the final reference iterate and the iterate $250$ steps earlier for instance $i$. We display only tested matrix tolerances satisfying
\begin{equation}\label{eq:displayed-matrix-tolerance}
 \varepsilon\ge\max\left\{10^{-12},\;10\max_{1\le i\le500}d_i\right\}.
\end{equation}
This condition provides a heuristic check of reference stability. It does not bound the remaining distance between $X^{\mathrm{ref}}$ and the limiting matrix $X^{(\infty)}$. The reverse triangle inequality gives
\begin{equation}\label{eq:reference-matrix-error}
\begin{aligned}
 &\left|\fnorm{X^{(k)}-X^{\mathrm{ref}}}
 -\fnorm{X^{(k)}-X^{(\infty)}}\right|\\
 &\qquad\le\fnorm{X^{\mathrm{ref}}-X^{(\infty)}}.
\end{aligned}
\end{equation}
The bottom row therefore reports distance to the numerical reference, without certifying the distance to the limiting matrix.

Continuing each iteration for $1000$ steps beyond its reference iterate changed the recorded matrix iteration counts by less than $0.05\%$. This provides an empirical check of stability at the plotted tolerances.

\begingroup\linespread{0.98}\selectfont
\endgroup
\end{document}